\documentclass[journal]{IEEEtran}

\usepackage[noadjust]{cite}

\usepackage{mathtools} 
\usepackage{amsmath}
\usepackage{amssymb}
\usepackage{mathtools}
\usepackage{amsthm}
\usepackage{xcolor}
\usepackage{balance}
\usepackage{subcaption}
\usepackage{float}
\usepackage{bbm}
\usepackage[ruled,vlined]{algorithm2e}
\SetKwInput{KwInput}{Input}
\SetKwInput{KwOutput}{Output}
\SetKwComment{Comment}{/*}{ */}
\let\oldnl\nl
\newcommand{\nonl}{\renewcommand{\nl}{\let\nl\oldnl}}
\newcommand{\resp}[1]{#1}
\newcommand{\rev}[1]{#1}

\theoremstyle{plain}
\newtheorem{theorem}{Theorem}[section]

\newtheorem{lemma}[theorem]{Lemma}

\theoremstyle{definition}
\newtheorem{definition}[theorem]{Definition}
\newtheorem{assumption}[theorem]{Assumption}
\theoremstyle{remark}
\newtheorem{remark}[theorem]{Remark}

\def\de{=}

\DeclareMathOperator*{\argmin}{arg\,min}
\newcommand\eqnum{\addtocounter{equation}{1}\tag{\theequation}}

\onecolumn
\begin{document}

\title{Robust Score-Based Quickest Change Detection}

\author{Sean~Moushegian,
        Suya~Wu,
        Enmao~Diao,
        Jie~Ding,
        Taposh Banerjee,
        and~Vahid~Tarokh

\thanks{Sean Moushegian, Suya Wu, Enmao Diao, and Vahid Tarokh are with the Department of Electrical and Computer Engineering, 
Duke University, Durham, North Carolina 27708.    Jie Ding is with the School of Statistics, University of Minnesota Twin Cities, Minneapolis, MN 55455.  Taposh Banerjee is with the Department of Industrial Engineering,
University of Pittsburgh, Pittsburgh, Pennsylvania 15213.  

An earlier version of this paper was presented at the 39th Conference on Uncertainty in Artificial Intelligence.

This paper was accepted to IEEE Transactions on Information Theory: 
S. Moushegian, S. Wu, E. Diao, J. Ding, T. Banerjee and V. Tarokh, ``Robust Score-Based Quickest Change Detection," in IEEE Transactions on Information Theory, vol. 71, no. 7, pp. 5539-5555, July 2025, doi: 10.1109/TIT.2025.3566677, available at https://doi.org/10.1109/TIT.2025.3566677 and https://ieeexplore.ieee.org/document/10982143.

© 2025 IEEE. Personal use of this material is permitted. Permission from IEEE must be obtained for all other uses, in any current or future media, including reprinting/republishing this material for advertising or promotional purposes, creating new collective works, for resale or redistribution to servers or lists, or reuse of any copyrighted component of this work in other works.  See https://www.ieee.org/publications/rights/index.html for more information.
}}

\maketitle


\begin{abstract}
   Methods in the field of quickest change detection rapidly detect in real-time a change in the data-generating distribution of an online data stream.  Existing methods have been able to detect this change point when the densities of the pre- and post-change distributions are known.  Recent work has extended these results to the case where the pre- and post-change distributions are known only by their score functions.  This work considers the case where the pre- and post-change score functions are known only to correspond to distributions in two disjoint sets.  This work selects a pair of least-favorable distributions from these sets to robustify the existing score-based quickest change detection algorithm, the properties of which are studied.  This paper calculates the least-favorable distributions for specific model classes and provides methods of estimating the least-favorable distributions for common constructions.  Simulation results are provided demonstrating the performance of our robust change detection algorithm.
\end{abstract}


\IEEEpeerreviewmaketitle
 
\section{Introduction}

\IEEEPARstart{I}{n}
the fields of sensor networks, cyber-physical systems, biology, and neuroscience, the statistical properties of online data streams can suddenly change in response to some application-specific event (\cite{veeravalli2014quickest,  basseville1993detection, poor2008quickest, tartakovsky2014sequential}).  The field of quickest change detection aims to detect the change in the underlying distribution of this observed stochastic process as rapidly as possible and in real-time -- but to do so with minimal risk of detecting a change before it actually occurs.  When the pre- and post-change probability density functions of the data are known, three important algorithms in the literature are the Shiryaev algorithm (\cite{shiryaev1963optimum, tartakovsky2005general}), the cumulative sum (CUSUM) algorithm (\cite{lorden1971procedures, moustakides1986optimal, lai1998information, page1955test}), and the Shiryaev-Roberts algorithm (\cite{pollak1985optimal, roberts1966comparison}).  These three tests calculate a sequence of statistics using the likelihood ratio of the observations and detection occurs when the statistics exceed a threshold (see \cite{shiryaev1963optimum,tartakovsky2005general,lorden1971procedures,moustakides1986optimal,lai1998information,pollak1985optimal}).

The main challenge in implementing a change detection algorithm in practice is that the pre- and post-change distributions are often not precisely known. This challenge is amplified when the data is high-dimensional. Specifically, in several machine learning applications, the data models may not lend themselves to explicit distributions. For example, energy-based models~(\cite{LeCun2006ATO}) capture dependencies between observed and latent variables based on their associated energy (an unnormalized probability), and score-based deep generative models~\cite{song2020score} generate high-quality images by learning the score function (the gradient of the log density function). These models can be computationally cumbersome to normalize as probabilistic density functions. Thus, optimal algorithms from the change detection literature, which are likelihood ratio-based tests, are computationally expensive to implement. 

This issue is partially addressed in \cite{wu_IT_2024} which proposed the SCUSUM algorithm, a Hyv\"arinen score-based (\cite{hyvarinen2005estimation}) modification of the CUSUM algorithm for quickest change detection. 
In \cite{wu_IT_2024} it is shown that the SCUSUM algorithm is consistent, and expressions for the average detection delay and the mean time to a false alarm of SCUSUM are provided. The 
Hyv\"arinen score is invariant to scale and hence can be applied to unnormalized models. This makes the SCUSUM algorithm highly efficient as compared to the classical CUSUM algorithm for high-dimensional models.

The main limitation of the SCUSUM algorithm is that its effectiveness is contingent on knowing the precise pre- and post-change unnormalized or score-based models, i.e., knowing the pre- and post-change models within a normalizing constant or knowing the gradient of log density. In practice, due to a limited amount of training data, the models can only be learned within an uncertainty class. To detect the change effectively, an algorithm must be robust against these modeling uncertainties. The SCUSUM algorithm is not robust in this sense. Specifically, if not carefully designed, the SCUSUM algorithm can fail to detect several (in fact, infinitely many) pre- and post-change scenarios. 

As one concrete example, consider a device that monitors patient health in a hospital setting.  The device observes a time series of biomedical data (e.g. patient heart rate) and detects when a patient changes from stable to unstable health condition.  While it is possible to learn general patterns between patient stability and heart rate, it is also true that every patient has a unique resting heart rate.  An SCUSUM algorithm trained on stable and unstable patient data will suffer premature false alarms or significant detection delays when it operates on a patient with an unusual resting heart rate due to this lack of robustness.

In this paper, a robust score-based variant of the CUSUM algorithm, called RSCUSUM, is proposed for the problem of quickest change detection.  Under the assumption that the pre- and post-change uncertainty classes are disjoint and convex, we show that the algorithm is robust, i.e., can consistently detect changes for every possible pre- and post-change model. This consistency is achieved by designing the RSCUSUM algorithm using a pair of distributions from the pre- and post-change uncertainty classes that are least favorable in a well-defined sense.

The problem of optimal robust quickest change detection is studied in \cite{unnikrishnan2011minimax}. In a minimax setting, the optimal algorithm is the CUSUM algorithm designed using the least favorable distributions. 
The robust CUSUM test in \cite{unnikrishnan2011minimax} may suffer from two drawbacks: 1) It is a likelihood ratio-based test and hence may not be amenable to implementation in high-dimensional models. 2) The notion of least favorable distribution is defined using {stochastic boundedness} (\cite{unnikrishnan2011minimax}), which may be difficult to verify for high-dimensional data. 

In contrast with the work in \cite{unnikrishnan2011minimax}, we define the notion of least favorable distribution using Fisher divergence and provide a method to identify the least favorable distributions.

We now summarize our contributions in this paper. 
\begin{enumerate}
\item We propose a new robust score-based quickest change detection (RSCUSUM) algorithm that can be applied to score-based models and unnormalized models. 
\resp{As with the SCUSUM method from \cite{wu_IT_2024}}, the role of Kullback-Leibler divergence in classical change detection is replaced with the Fisher divergence between the pre- and post-change distributions\resp{, and the role of log-likelihood ratios is replaced with differences between  Hyv\"arinen scores (\cite{hyvarinen2005estimation})}. This algorithm is introduced in Section~\ref{sec:RSCUSUM_algorithm}.

\item The RSCUSUM algorithm can address unknown pre- and post-change models. Specifically, assuming that the pre- and post-change laws belong to known families of distributions that are convex, we identify a least favorable pair of distributions that are nearest in terms of Fisher divergence. We then show that the RSCUSUM algorithm can consistently detect each post-change distribution from each pre-change distribution from the relevant families, and is robust in this sense. This analysis is provided in Section~\ref{sec:theoritical_analysis}.

\item 
We provide an effective method to identify the least favorable distributions among pre- and post-change families. This is in contrast to the setup in \cite{unnikrishnan2011minimax}, where a stochastic boundedness characterization makes it harder to identify the least favorable distributions. This result is presented in Section~\ref{sec:least_favorable_distribution}.


\item 
\resp{As with the SCUSUM algorithm, the RSCUSUM algorithm is a score-based algorithm where cumulative scores do not enjoy a standard martingale characterization.  This differentiates the RSCUSUM algorithm and the SCUSUM algorithm from the classical CUSUM algorithm, which leverages the fact that the likelihood ratios form a martingale under the pre-change model~\cite{lai1998information,woodroofe1982nonlinear}.  We generalize the analysis of the delay and false alarm of the SCUSUM method to the robust case.  This analysis is presented in Section~\ref{sec:theoritical_analysis}.}

\item We identify the least-favorable distribution for a specific case involving Gaussian mixture models.  We then demonstrate the effectiveness of the RSCUSUM algorithm through simulation studies on a Gaussian mixture model and a Gauss-Bernoulli Restricted Boltzmann Machine (\cite{gbrbm_without_tears}). The identification of the least-favorable distributions for this model class is presented in Section~\ref{sec:least_favorable_distribution} with simulation results presented in Section~\ref{sec:results}. 
\end{enumerate}

The remainder of this paper is organized as follows:  in Section~\ref{subsec:problem_formulation}, the problem of change-point detection is formally defined.  In Section~\ref{sec:soln_densities}, we review a detection algorithm in the case where the densities of the pre- and post-change distributions are known.  In Section~\ref{sec:soln_scores}, we review an algorithm that detects change points when we have knowledge of the scores of the pre- and post-change distributions without knowledge of their densities.

\section{Problem Formulation}

\label{subsec:problem_formulation}
\noindent Let $\{X_n\}_{n\geq 1}$ denote a sequence of independent random variables defined on the probability space $(\Omega, \mathcal{F}, \mu_\nu)$. Let $\mathcal{F}_n$ be the $\sigma-$algebra generated by random variables $X_1,\; X_2, \;\dots,\; X_n$, and let $\mathcal{F}=\sigma(\cup_{n\geq 1}\mathcal{F}_n)$ be the $\sigma-$algebra generated by the union of sub-$\sigma$-algebras $\{\mathcal{F}_n$\}. 
Under $\mu_\nu$, $X_1, \; X_2, \;\dots,\; X_{\nu-1}$ are independent and identically distributed (i.i.d.) according to probability measure $P_\infty$ (with density $p_\infty$) and $X_{\nu}, \; X_{\nu+1},\; \dots$ are {i.i.d.} according to probability measure $P_1$ (with density $p_1$). We think of $\nu$ as the change-point, $p_\infty$ as the pre-change density, and $p_1$ as the post-change density. We use $\mathbb{E}_{\nu}$ and $\text{Var}_{\nu}$ to denote the expectation and the variance associated with the
measure $\mu_\nu$, respectively. Thus, $\nu$ is seen as an unknown constant and we have an entire family $\{\mu_\nu\}_{1 \leq \nu \leq \infty}$ of change-point models, one for each possible change-point. We use $\mu_\infty$ to denote the measure under which there is no change, with $\mathbb{E}_\infty$ denoting the corresponding expectation.

A change detection algorithm is a stopping time $T$ with respect to the data stream $\{X_n\}_{n\geq 1}$:
$$
\{T \leq n\} \in \mathcal{F}_n, \quad \forall n \geq 1. 
$$
If $T\geq \nu$, we have made a {delayed detection}; otherwise, a {false alarm} has happened. 
Our goal is to find a stopping time $T$ to optimize the trade-off between well-defined metrics on delay and false alarm. 
We consider two minimax problem formulations to find the best stopping rule.

To measure the detection performance of a stopping rule, we use the following minimax metric (\cite{lorden1971procedures}), the worst-case averaged detection delay (WADD):
\begin{equation*}
\mathcal{L}_{\texttt{\textup{WADD}}}(T)\de \sup_{\nu\geq 1}\text{ess}\sup \mathbb{E}_{\nu}[(T-\nu+1)^{+}|\mathcal{F}_{\nu-1}],
\end{equation*}
where $(y)^{+}\de\max(y, 0)$ for any $y\in \mathbb{R}$. Here $\text{ess} \sup$ is the essential supremum, i.e., the supremum outside a set of measure zero. We also consider the version of minimax metric introduced in \cite{pollak1985optimal}, the worst conditional averaged detection delay (CADD):
\begin{equation*}
    \mathcal{L}_{\texttt{\textup{CADD}}}(T)\de \sup_{\nu\geq 1}\mathbb{E}_{\nu}[T-\nu|T\geq \nu].
\end{equation*}
For false alarms, we consider the {average running length} (ARL), which is defined as the mean time to false alarm:
$$
\text{ARL}\resp{(T)}\de \mathbb{E}_\infty[T].
$$

We assume that pre- and post-change distributions are not precisely known. However, each is known within an uncertainty class. Let $\mathcal{G}_\infty$ and $\mathcal{G}_1$ be two disjoint classes of probability measures. Then, we assume that the distributions $P_\infty$ and $P_1$ satisfy
\begin{equation}
\label{eq:UncertainityClasses_t}
    \begin{split}
        P_\infty &\in \mathcal{G}_\infty, \\
        P_1 &\in \mathcal{G}_1. 
    \end{split}
\end{equation}
The objective is to find a stopping rule to solve the following problem:
\begin{equation}
    \label{eq:lorden}
    \min_T \; \resp{\sup_{(P_\infty, P_1) \in \mathcal{G}_\infty \times \mathcal{G}_1}} \mathcal{L}_{\texttt{\textup{WADD}}}(T)\; 
    \; \text{ s.t.}\; \resp{\inf_{P_\infty \in \mathcal{G}_\infty }}\mathbb{E}_\infty[T]\geq \gamma,
\end{equation}
 where $\gamma$ is a constraint on the ARL. Precise assumptions on the families $\mathcal{G}_1$ and $\mathcal{G}_\infty$ will be discussed below. 
 
 The delay $\mathcal{L}_{\texttt{\textup{WADD}}}$ in the above problem is a function of the true pre- and post-change laws $P_\infty, P_1$ and should be designated as
 $
\mathcal{L}_{\texttt{\textup{WADD}}}^{P_\infty, P_1},
 $ 
 but we suppress the notation and refer to this quantity simply as $\mathcal{L}_{\texttt{\textup{WADD}}}$.
 Similarly, we will suppress the dependence of $\mathbb{E}_\infty[T]$ on the pre-change law $P_\infty$. 
 Thus, the goal in this problem is to find a stopping time $T$ to minimize the worst-case detection delay, subject to a constraint $\gamma$ on $\mathbb{E}_\infty[T]$. 
 
 We are also interested in the minimax formulation: 
\begin{equation}
    \label{eq:pollak}
    \min_T \;\resp{\sup_{(P_\infty, P_1) \in \mathcal{G}_\infty \times \mathcal{G}_1}} \mathcal{L}_{\texttt{\textup{CADD}}}(T) 
    \; \text{ s.t. }\; \resp{\inf_{P_\infty \in \mathcal{G}_\infty }}\mathbb{E}_\infty[T]\geq \gamma. 
\end{equation}
The delay $\mathcal{L}_\text{CADD}$ is also a function of $P_\infty, P_1$, but we suppress the notation and refer to $\mathcal{L}_\text{CADD}^{P_\infty, P_1}$ as $\mathcal{L}_\text{CADD}$. For an overview of the literature on quickest change detection when the densities are not precisely known, we refer the readers to \cite{veeravalli2014quickest,  tartakovsky2014sequential, basseville1993detection, lai1998information, lorden1971procedures}.

\section{Optimal Solution Based on Likelihood Ratios}
\label{sec:soln_densities}
If both the pre-change and post-change families are singletons, $\mathcal{G}_1 = \{P_1\}, \mathcal{G}_\infty = \{ P_\infty\}$, and $P_1 \neq P_\infty$, then the above formulations are the classical minimax formulations from the quickest change detection literature; see \cite{veeravalli2014quickest, poor2008quickest, tartakovsky2014sequential}. The optimal algorithm (exactly optimal for \eqref{eq:lorden} and asymptotically optimal for \eqref{eq:pollak}) is the CUSUM algorithm given by
\begin{equation} 
\label{eq:baseline}
    T_{\texttt{CUSUM}} \de \inf\{n\geq 1:\Lambda(n)\geq \omega\},
\end{equation}
where $\Lambda(n)$ is defined using the recursion
\begin{equation}
\begin{split}
\label{eq:cusum_score}
    &\Lambda(0)\de0, \nonumber \\
    &\Lambda(n) \de \biggr(\Lambda(n-1)+\log \frac{p_1(X_n)}{p_{\infty}(X_n)}\biggr)^{+}, \;\quad 
    \forall n \geq 1, 
    \end{split}
 \end{equation}
which leads to a computationally convenient stopping scheme. We recall that here $p_1$ is the post-change density and $p_\infty$ is the pre-change density. In \cite{lorden1971procedures} and \cite{lai1998information}, the asymptotic performance of the CUSUM algorithm is also characterized. Specifically, it is shown that with $\omega = \log \gamma$, 
$$
\mathbb{E}_\infty[T_{\texttt{CUSUM}}] \geq \gamma,
$$
and as $\gamma \rightarrow \infty$,
\begin{align*}
    \mathcal{L}_{\texttt{\textup{WADD}}}(T_{\texttt{CUSUM}}) \sim \mathcal{L}_{\texttt{\textup{CADD}}}(T_{\texttt{CUSUM}})\sim \frac{\log \gamma}{\mathbb{D}_{\texttt{\textup{KL}}}(P_{1}\|P_{\infty})}.
\end{align*}
Here $\mathbb{D}_{\texttt{\textup{KL}}}(P_{1}\|P_{\infty})$ is the Kullback-Leibler divergence between the post-change distribution and the pre-change distribution:
$$
\mathbb{D}_{\texttt{\textup{KL}}}(P_{1}\|P_{\infty}) \de \int p_1(x) \log \frac{p_1(x)}{p_\infty(x)} dx, 
$$
and the notation $g(c)\sim h(c)$ as $c\to c_0$ indicates that ${g(c)}/{h(c)} \to 1$ as $c\to c_0$ for any two functions $c\mapsto g(c)$ and $c\mapsto h(c)$.

If the densities $p_\infty, p_1$ are not known and are assumed to belong to families $\mathcal{G}_\infty, \mathcal{G}_1$, i.e., the families $\mathcal{G}_\infty, \mathcal{G}_1$ are not singleton sets, then the optimal test is designed using the least favorable distributions in a stochastic bounded sense. Specifically, in \cite{unnikrishnan2011minimax}, it is assumed that there are densities $\bar{p}_\infty \in \mathcal{G}_\infty, \bar{p}_1 \in \mathcal{G}_1$ such that for every $p_\infty \in \mathcal{G}_\infty$ and $p_1 \in \mathcal{G}_1$, 
\begin{equation}
\label{eq:leastfavunni}
    \begin{split}
        \log \frac{\bar{p}_1(X)}{\bar{p}_\infty(X)} \bigg|_{X \sim \bar{p}_1} \; \; \prec \quad \; \; \log \frac{\bar{p}_1(X)}{\bar{p}_\infty(X)}\bigg|_{X \sim p_1}  ,
    \end{split}
\end{equation}
\begin{equation}
\label{eq:leastfavunni_post}
    \begin{split}
        \log \frac{\bar{p}_1(X)}{\bar{p}_\infty(X)}\bigg|_{X \sim \bar{p}_\infty} 
        \; \; 
        \succ 
        \quad \; \; 
        \log \frac{\bar{p}_1(X)}{\bar{p}_\infty(X)} \bigg|_{X \sim p_\infty} 
        . 
    \end{split}
\end{equation}
Here the notation $\prec$ is used to denote stochastic dominance: if $W$ and $Y$ are two random variables, then $W \prec Y$ if
$$
P(Y \geq t) \geq P(W \geq t), \quad \text{for all } t \in (-\infty, \infty). 
$$
If such densities $\bar{p}_\infty, \bar{p}_1$ exist in the pre- and post-change families, then the optimal algorithm is the robust CUSUM algorithm designed with $\bar{p}_\infty, \bar{p}_1$ used as the pre- and post-change densities: 
\begin{equation*}
\label{eq:cusum_robust_LFD}
    \bar{\Lambda}(0)=0,  \;\quad\quad
    \bar{\Lambda}(n) \de \biggr(\bar{\Lambda}(n-1)+\log \frac{\bar{p}_1(X_n)}{\bar{p}_{\infty}(X_n)}\biggr)^{+}, \; \forall n \geq 1.
 \end{equation*}
The optimality in \cite{unnikrishnan2011minimax} is established under additional assumptions on the smoothness of densities. We refer the readers to \cite{unnikrishnan2011minimax} for a more precise optimality statement. 

We note that in the literature on quickest change detection, the issue of the unknown post-change model has been addressed by using a generalized likelihood ratio (GLR) test or a mixture-based test. While these tests have strong optimality properties, they are computationally expensive to implement; see \cite{lorden1971procedures, lai1998information, tartakovsky2014sequential}.

\section{Quickest Change Detection in Unnormalized and Score-Based Models}
\label{sec:soln_scores}
The limitation of the CUSUM and the robust CUSUM algorithms is that they are based on likelihood ratios which are not always precisely known. In modern machine learning applications, two new classes of models have emerged:
\begin{enumerate}
\item \textit{Unnormalized statistical models}: In these models, we know the distribution within a normalizing constant:
\begin{equation}
\begin{split}
    p_\infty(x) = \frac{\tilde{p}_\infty(x)}{Z_\infty},  \quad \text{ and } \quad 
    p_1(x) = \frac{\tilde{p}_1(x)}{Z_1},
\end{split}
\end{equation}
where $Z_\infty$ and $Z_1$ are normalizing constants that are hard (or even impossible) to calculate by numerical integration, especially when the dimension of $x$ is high. 
In some applications, the unnormalized models $\tilde{p}_\infty(x)$ and $\tilde{p}_1(x)$ are known in precise functional forms. Examples include continuous-valued Markov random fields or undirected graphical models which are used for image modeling. We refer the reader to \cite{hyvarinen2005estimation, wu_IT_2024} for detailed discussions on unnormalized models.   
 \item \textit{Score-based models}: Often, even $\tilde{p}_\infty(x)$ and $\tilde{p}_1(x)$ are unknown. But, it may be possible to learn the scores
	$$
	\nabla_x \log p_\infty(x), \quad \text{ and 
 }\quad \nabla_x \log p_1(x),
	$$
 from data. Here $\nabla_{x}$ is the gradient operator. This is possible using the idea of {score-matching}. Specifically, these scores can be learned using a deep neural network. We refer the readers to 
 \cite{hyvarinen2005estimation, song2020score, song2019score, vincent2011connection, wu_IT_2024, wu2022score} for details. We note that a score-based model is also unnormalized where the exact form of the unnormalized function is hard to estimate. 
\end{enumerate}

In \cite{wu_IT_2024}, a score-based CUSUM (SCUSUM) algorithm was developed to detect changes in unnormalized and score-based models and performance characteristics of this algorithm were obtained.  
In the score-based theory in \cite{wu_IT_2024}, the Kullback-Leibler divergence is replaced by the Fisher divergence  (defined precisely below) between the pre- and post-change densities. 

The SCUSUM algorithm is defined based on Hyv\"arinen Score (\cite{hyvarinen2005estimation}), which circumvents the computation of normalization constants and has found diverse applications including Bayesian model comparison~\cite{Shao2019}, information theory~\cite{ding2019gradient}, and hypothesis testing~\cite{wu2022score}.
We first define the Hyv\"arinen Score below.
\begin{definition}[Hyv\"arinen Score] The Hyv\"arinen score of any probability measure $P$ (with density $p$) is a mapping $(x, P)\mapsto \mathcal{S}_{\texttt{\textup{H}}}(x, P)$ given by 
    \begin{equation*}
        \mathcal{S}_{\texttt{\textup{H}}}(x, P) \de \frac{1}{2} \left \| \nabla_{x} \log p(x) \right \|_2^2 + \Delta_{x} \log p(x),
    \end{equation*}
whenever it can be well defined, where 
$$\Delta_{x} \de \sum_{i=1}^d \frac{\partial^2}{\partial x_i^2}$$
and $\nabla_x$ are the Laplacian and gradient operators acting on $X = (x_1, \cdots, x_d)^{\top}$, and where $\|\cdot\|_2$ denotes the Euclidean norm.  We assume that the Hyv\"arinen score is well defined for all $P \in \mathcal{G}_\infty \cup \mathcal{G}_1$ by Assumption~\ref{assumption:regularity}.
\end{definition}

By using the Hyv\"arinen Score in our algorithm, the role of Kullback-Leibler divergence in the theoretical analysis of the algorithm is replaced by the Fisher divergence. 

\begin{definition}[Fisher Divergence] The Fisher divergence between two probability measures $P$ to $Q$ (with densities $p$ and $q$) is defined by
\begin{equation}\label{eq:FisherDivDef}
    \mathbb{D}_{\texttt{\textup{F}}} (P \| Q) \de \mathbb{E}_{X\sim P} \left[ \rev{\frac{1}{2}} \left\| \nabla_{{x}} \log p(X)- \nabla_{{x}} \log q(X)\right \|_2^2 \right],
\end{equation}
whenever the integral is well defined. 
\end{definition}

Clearly, $\nabla_{{x}} \log p(x)$, $\nabla_{{x}} \log q(x)$, and $\Delta_{x} \log q(x)$ remain invariant if $p$ and $q$ are scaled by any positive constant with respect to $x$. Hence, the Fisher divergence and the Hyv\"arinen Score remain {scale-invariant} concerning an arbitrary constant scaling of density functions.

To design the SCUSUM algorithm, the precise knowledge of probability measures $P_\infty$ and $P_1$ is not required. It is enough if we know the densities in unnormalized form or know their scores. Specifically, if we can calculate the Hyv\"arinen scores $\mathcal{S}_{\texttt{\textup{H}}}(x, P_{\infty})$ and $\mathcal{S}_{\texttt{\textup{H}}}(x, P_{1})$, then we can define the SCUSUM algorithm as\footnote{In fact, the SCUSUM algorithm of \cite{wu_IT_2024} is scaled by a positive constant $\lambda$.  We discuss the unscaled SCUSUM algorithm and introduce similar scaling by $\rho$ in Section~\ref{sec:theoritical_analysis}.} 
\begin{equation} 
\label{eq:SCUSUM_rule_old}
    T_{\texttt{\textup{SCUSUM}}} \de \inf\{n\geq 1:Y(n)\geq \omega\},
\end{equation}
where $\omega>0$ is a stopping threshold that is pre-selected to control false alarms, and \resp{$Y(n)$} can be computed recursively:
\begin{align*}
    &Y(0) \de 0, \\
    &Y(n) \de (Y(n-1)+\mathcal{S}_{\texttt{\textup{H}}}(X_n, P_{\infty})-\mathcal{S}_{\texttt{\textup{H}}}(X_n, P_{1}))^{+},\;\forall n\geq 1.
\end{align*}

\resp{We now impose the following conditions on the classes $\mathcal{G}_\infty$ and $\mathcal{G}_1$ in order to make claims regarding the SCUSUM Algorithm. All lemmas and theorems in the paper, unless otherwise stated, will be assumed to be valid when all these conditions are satisfied. }
\begin{assumption} \label{assumption:disjointconvex}
\resp{$\mathcal{G}_\infty, \mathcal{G}_1$ are disjoint and are each convex.}
\end{assumption}

\begin{assumption} \label{assumption:regularity}
\resp{Following regularity conditions introduced in \cite{hyvarinen2005estimation} hold: 
\begin{enumerate}
    \item For all densities $g \in \mathcal{G}_\infty \cup \mathcal{G}_1$, $g$ is differentiable.
    \item For all $g \in \mathcal{G}_\infty \cup \mathcal{G}_1$, $\mathbb{E}_{x \sim g} [\| \nabla_x \log g(x)\|^2_2] < \infty$.
    \item For all $g \in \mathcal{G}_\infty \cup \mathcal{G}_1$ and for all $h \in \mathcal{G}_\infty \cup \mathcal{G}_1$, $g(x) \nabla_x \log h(x) \rightarrow 0$ as $\| x\| \rightarrow \infty$.
\end{enumerate}}
\end{assumption}

\begin{assumption} \label{assumption:support}
\resp{All $P \in \mathcal{G}_\infty \cup \mathcal{G}_1$ share the same support $\mathcal{X}$.}
\end{assumption}

\resp{Note that the regularity conditions contained in Assumption~\ref{assumption:regularity}  were assumed in \cite{hyvarinen2005estimation} to ensure that the Hyv\"arinen score is well-defined. We also assume that all the moment-generating functions appearing in the paper are finite. } The following theorem is established in \cite{wu_IT_2024} regarding the consistency and the performance of the SCUSUM algorithm.

\begin{theorem}[\cite{wu_IT_2024}]
\label{thm:SCUSUM_IT_2024}
The following statements are true for the SCUSUM algorithm: 

    \textit{(a) Consistency:} The SCUSUM algorithm is consistent. Specifically, the drift of the statistic is negative before the change and is positive after the change:
\begin{equation}
    \begin{split}
\mathbb{E}_{\infty}\left[\mathcal{S}_{\texttt{\textup{H}}}(X_1, P_{\infty})-\mathcal{S}_{\texttt{\textup{H}}}(X_1, P_{1})\right] \; &= \; -\mathbb{D}_{\texttt{\textup{F}}}(P_\infty \; \| \;  P_1)<0,\\
\mathbb{E}_{1}\left[\mathcal{S}_{\texttt{\textup{H}}}(X_1, P_{\infty})-\mathcal{S}_{\texttt{\textup{H}}}(X_1, P_{1})\right] \; &= \mathbb{D}_{\texttt{\textup{F}}}(P_1 \; \| \;  P_\infty)> 0.
    \end{split}
\end{equation}

\textit{(b) Average Run Length:} Let there exist a $\lambda > 0$ such that 
\begin{equation}
\label{eq:lambdastar_old}
\mathbb{E}_{\infty}[\exp(\lambda (\mathcal{S}_{\texttt{\textup{H}}}(X_1, P_{\infty})-\mathcal{S}_{\texttt{\textup{H}}}(X_1, P_{1})))] \leq 1.
\end{equation}
Then, for any $\omega>0$,
    \begin{equation}
    \label{eq:arl}
        \mathbb{E}_{\infty}[T_{\texttt{\textup{SCUSUM}}}]\geq  e^{\lambda \omega}.
    \end{equation}
    Thus, setting $\omega={(\log \gamma)}/{\lambda}$ in 
  \eqref{eq:SCUSUM_rule_old} implies 
    $$
    \mathbb{E}_{\infty}[T_{\texttt{\textup{SCUSUM}}}]\geq \gamma.
    $$
    Thus, similar to the CUSUM algorithm, even the SCUSUM algorithm enjoys a universal lower bound on the mean time to false alarm for any distribution pair $(P_\infty, P_1)$. A $\lambda$ satisfying  \eqref{eq:lambdastar_old} always exists, otherwise, the problem is trivial.

    \textit{(c) Expected Detection Delay:}  Finally, with $\omega={(\log \gamma)}/{\lambda}$, the delay performance is given by 
    \begin{equation} 
  \mathcal{L}_{\texttt{\textup{WADD}}}(T_{\texttt{\textup{SCUSUM}}}) \sim \frac{\log \gamma}{\lambda \mathbb{D}_{\texttt{\textup{F}}}(P_{1} \; \|\; P_\infty)}, \quad \text{ as } \gamma \to \infty. 
    \end{equation}
Thus, the expected detection delay depends inversely on the Fisher divergence between $P_1$ and $P_\infty$. Thus, the role of KL-divergence in the classical CUSUM algorithm is replaced by the Fisher divergence in the score-based CUSUM algorithm. 

\end{theorem}

\section{Robust Quickest Change Detection in Unnormalized and Score-Based Models} \label{sec:RSCUSUM_algorithm}

In most applications, even the scores are not precisely known. For example, not enough training data may be available for precise score-matching. 
Thus, the SCUSUM algorithm cannot be applied. In this paper, we take a robust approach to address this issue. Specifically, we assume that both $\mathcal{G}_\infty$ and $\mathcal{G}_1$ are families of unnormalized or score-based models. In addition, we assume that there exists a pair of least favorable distributions $Q_\infty, Q_1$ in the following sense:

\begin{definition}[Least-favorable distributions (LFD)]
\label{def:lfd}
We say that a pair of distributions $(Q_\infty, Q_1)$ are least favorable if they are a solution to the following optimization problem: 
\begin{equation}
\label{eq:lfd}
\resp{(Q_\infty, Q_1) \de \argmin_{(R_\infty, R_1) \in \mathcal{G}_\infty \times \mathcal{G}_1} \mathbb{D}_{\texttt{\textup{F}}}(R_1 \| R_\infty).}
\end{equation}
\end{definition}
Thus, the distributions $Q_\infty \in \mathcal{G}_\infty, Q_1 \in \mathcal{G}_1$ are closest as measured through the pseudo distance of Fisher divergence (see \eqref{eq:FisherDivDef}). 
We remark that since the observations are assumed to be high-dimensional, the densities are not precisely known. As a consequence, the stochastic boundedness condition of \cite{unnikrishnan2011minimax} or the KL divergence cannot be used here to define the least favorable pair. On the other hand, the Fisher divergence can be computed for unnormalized and score-based models. 
In the rest of the paper, we assume that we can always find the least favorable pair $(Q_\infty, Q_1)$. In Section~\ref{sec:least_favorable_distribution}, we provide several examples where this assumption is valid. Finally, if multiple pairs are satisfying \eqref{eq:lfd}, then we assume that any one pair has been chosen arbitrarily. 

\resp{The motivation for using \eqref{eq:lfd} comes from the following analysis. Note that for any stopping rule $T$ such that 
$\inf_{P_\infty \in \mathcal{G}_\infty }\mathbb{E}_\infty[T]\geq \gamma$, we have as $\gamma \to \infty$, 
\begin{align*}
    \label{eq:lorden_tb}
    \min_T \; \sup_{(P_\infty, P_1) \in \mathcal{G}_\infty \times \mathcal{G}_1}\mathcal{L}_{\texttt{\textup{WADD}}}(T) 
    &  \geq \sup_{(P_\infty, P_1) \in \mathcal{G}_\infty \times \mathcal{G}_1} \min_T \;\mathcal{L}_{\texttt{\textup{WADD}}}(T) \\
    &  \geq \frac{\log \gamma}{\inf_{(P_\infty, P_1) \in \mathcal{G}_\infty \times \mathcal{G}_1}\mathbb{D}_{\texttt{\textup{KL}}}(P_1 \| P_\infty)}(1+o(1)), \eqnum
\end{align*}
where the $o(1)$ terms goes to zero as $\gamma \to \infty$. Thus,
$$
\frac{\log \gamma}{\inf_{(P_\infty, P_1) \in \mathcal{G}_\infty \times \mathcal{G}_1}\mathbb{D}_{\texttt{\textup{KL}}}(P_1 \| P_\infty)}
$$
provides a lower bound on the asymptotic performance of any algorithm. This lower bound can be achieved by a CUSUM-type algorithm $T^*$ if $\inf_{P_\infty \in \mathcal{G}_\infty }\mathbb{E}_\infty[T^*]\geq \gamma$ and the drift of this algorithm after the change is 
$\inf_{(P_\infty, P_1) \in \mathcal{G}_\infty \times \mathcal{G}_1}\mathbb{D}_{\texttt{\textup{KL}}}(P_1 \| P_\infty)$. Since we cannot calculate the KL divergence for unnormalized and score-based models, we use the Fisher divergence as a proxy measure. We note also that these two divergences coincide for certain Gaussian families \cite{wu_IT_2024}. 
}

We now use $Q_\infty, Q_1$ and their densities $q_\infty, q_1$ to design a robust score-based cumulative sum (RSCUSUM) algorithm. 
We define the instantaneous RSCUSUM score function $x\mapsto z(x)$ by 
\begin{equation}
\label{eq:scusum_instantaneous}
    z(x) \de \mathcal{S}_{\texttt{\textup{H}}}(x, Q_{\infty})-\mathcal{S}_{\texttt{\textup{H}}}(x, Q_{1}),
\end{equation}
where $\mathcal{S}_{\texttt{\textup{H}}}(x, Q_{\infty})$ and $\mathcal{S}_{\texttt{\textup{H}}}(x, Q_1)$ are respectively the Hyv\"arinen score functions of $Q_\infty$ and $Q_1$. 
Our proposed stopping rule is given by 
\begin{equation} 
\label{eq:SCUSUM_rule}
    T_{\texttt{\textup{RSCUSUM}}} \de \inf\{n\geq 1:Z(n)\geq \omega\},
\end{equation}
where $\omega>0$ is a stopping threshold that is pre-selected to control false alarms, and $Z(n)$ can be computed recursively:
\begin{align*}
    &Z(0) \de 0, \;\quad\quad Z(n) \de (Z(n-1)+z(X_n))^{+},\;\forall n\geq 1.
\end{align*}
The statistic $Z(n)$ is referred to as the detection score of RSCUSUM at time $n$. The RSCUSUM algorithm is summarized in Algorithm~\ref{algm:rscusum}. The purpose of the rest of the paper is to establish conditions under which the RSCUSUM algorithm is consistent, and also provide its delay and false alarm analysis. 

\begin{algorithm}
\DontPrintSemicolon
\caption{RSCUSUM Quickest Change Detection Algorithm}
\label{algm:rscusum}
\KwInput{Hyv\"arinen score functions 
$\mathcal{S}_{\texttt{\textup{H}}}(\cdot, Q_{\infty})$ and $\mathcal{S}_{\texttt{\textup{H}}}(\cdot, Q_{1})$ of least favorable distributions in 
$\mathcal{G}_\infty, \mathcal{G}_1$, respectively.} 
\SetKwProg{Fn}{Initialization}{:}{}
  \Fn{}{
       Current time $k=0$,  $\omega>0$, and $Z(0)=0$}
\While{$Z(k)<\omega$}{
$k = k+1$\\
Update $z(X_k) = \mathcal{S}_{\texttt{\textup{H}}}(X_{k}, Q_{\infty})-\mathcal{S}_{\texttt{\textup{H}}}(X_{k}, Q_{1})$\\
Update $Z(k) = \max(Z(k-1)+z(X_k), 0)$\;
}
Record the current time $k$ as the stopping time $T_{\texttt{\textup{RSCUSUM}}}$\;
\KwOutput{$T_{\texttt{\textup{RSCUSUM}}}$}
\end{algorithm}

\section{Consistency and Delay and False Alarm Analysis of the RSCUSUM Algorithm} \label{sec:theoritical_analysis}
\noindent In this section, we prove the consistency (the ability to detect the change with probability one) and provide delay and false alarm analysis of the RSCUSUM algorithm. In Section~\ref{sec:RevTriangle}, we prove an important lemma that can be interpreted as a reverse triangle inequality for Fisher divergence. This lemma is then used in Section~\ref{sec:consistency} to prove the consistency of the RSCUSUM algorithm. In Section~\ref{sec:MFA}, we obtain a lower bound on the mean time to a false alarm, and in Section~\ref{sec:DelayAnalysis}, we obtain an expression for the average detection delay of the RSCUSUM algorithm. 

We first make another fundamental assumption: 
\begin{assumption} \label{assumption:nearness}
For the least favorable distributions $Q_\infty, Q_1$ of Definition~\ref{def:lfd}, $\mathbb{D}_{\texttt{\textup{F}}}(P_\infty \| Q_\infty) < \mathbb{D}_{\texttt{\textup{F}}}(P_\infty \| Q_1)$.
\end{assumption}
\begin{remark}
If $\mathcal{G}_\infty$ contains only one element ($\mathcal{G}_\infty = \{P_\infty \}$), then clearly the least favorable pre-change distribution $Q_\infty$ is equal to $P_\infty$ and $\mathbb{D}_{\texttt{\textup{F}}}(P_\infty \| Q_\infty) = 0$.  Then, Assumption~\ref{assumption:nearness} follows from the disjointness of $\mathcal{G}_\infty, \mathcal{G}_1$ given in Assumption~\ref{assumption:disjointconvex} and the regularity conditions given in Assumption~\ref{assumption:regularity}.
\end{remark}
Next, we provide a method of checking Assumption~\ref{assumption:nearness} for a specific construction of $\mathcal{G}_\infty$.
\begin{theorem}
\label{theorem:applicability}
Suppose we have a finite set of distributions:
\begin{align*}
    \rev{\mathcal{P}_m = \{P_i,\; i=1,\dots, m\},\; m\in \mathbb{N}^{+}}
\end{align*}
\resp{where each distribution $P_i$ has density $p_i$.  Further} suppose that $\mathcal{G}_\infty$ is defined to be the convex hull of this finite set:
\begin{equation}
    \mathcal{G}_\infty \de \biggl\{ x \mapsto \sum_{i=1}^m \alpha_i p_i(x):  \sum_{i=1}^m \alpha_i=1, \alpha_i \geq 0\biggr\}.
\end{equation}

If $\mathbb{D}_{\texttt{\textup{F}}}(P_i \| Q_\infty) < \mathbb{D}_{\texttt{\textup{F}}}(P_i \| Q_1)$ for all $1 \leq i \leq m$, then Assumption~\ref{assumption:nearness} holds for any $P_\infty \in \mathcal{G}_\infty$.
\end{theorem}
\begin{proof}
We use $C_R$ to denote the term $\mathbb{E}_{X\sim R} [\frac{1}{2}\left \| \nabla_{x} \log r(X) \right \|_2^2]$ for any distribution $R$. Then 
\begin{align*}
\begin{split}
     \mathbb{E}_{X \sim R}[z(X)] = \mathbb{E}_{X \sim R}[\mathcal{S}_{\texttt{\textup{H}}}(X, Q_{\infty})-\mathcal{S}_{\texttt{\textup{H}}}(X, Q_1)] 
     &=\mathbb{D}_{\texttt{\textup{F}}} (R \| Q_{\infty}) -C_{R} -\mathbb{D}_{\texttt{\textup{F}}} (R \| Q_1)+C_{R} \\
     &=\mathbb{D}_{\texttt{\textup{F}}} (R \| Q_\infty)-\mathbb{D}_{\texttt{\textup{F}}} (R \| Q_1),
     \end{split}
\end{align*}
where $z(x)$ is defined in Equation~\ref{eq:scusum_instantaneous}.  Note that $z(x)$ does not depend upon $P_\infty$ or $P_1$.

By the construction of $\mathcal{G}_\infty$, we can express $P_\infty = \alpha_1 P_1 + ...+ \alpha_m P_m$ for some (possibly unknown) $\alpha_1, ..., \alpha_m$:
\begin{align*}
    \mathbb{D}_{\texttt{\textup{F}}}(P_\infty  \| Q_\infty) - \mathbb{D}_{\texttt{\textup{F}}}(P_\infty \| Q_1) = &\mathbb{E}_{\infty}[z(X_1)]
    = \int p_\infty(x) z(x) dx = \int \sum_{i=1}^m \alpha_i p_i(x) z(x) dx \\
    & = \sum_{i=1}^m \alpha_i \int p_i(x) z(x) dx  = \sum_{i=1}^m \alpha_i \mathbb{E}_{P_i}[z(X)] \\
    &= \sum_{i=1}^m \alpha_i \bigg(\mathbb{D}_{\texttt{\textup{F}}}(P_i \| Q_\infty) - \mathbb{D}_{\texttt{\textup{F}}}(P_i \| Q_1)\bigg). \eqnum
\end{align*}
It is given that $\mathbb{D}_{\texttt{\textup{F}}}(P_i \| Q_\infty) - \mathbb{D}_{\texttt{\textup{F}}}(P_i \| Q_1) < 0$ for each $1 \leq i \leq m$ and that $\alpha_i \geq 0$ for all $1 \leq i \leq m$.  Thus, $\mathbb{D}_{\texttt{\textup{F}}}(P_\infty \| Q_\infty) - \mathbb{D}_{\texttt{\textup{F}}}(P_\infty \| Q_1)$ can be written as the sum of $m$ nonpositive terms, at least one of which is strictly negative.
\end{proof}

\subsection{Reverse Triangle Inequality for Fisher Divergence}\label{sec:RevTriangle}
We first prove an important lemma for our problem. Suppose the Fisher divergence is seen as a measure of distance between two probability measures. In that case, the following lemma provides a reverse triangle inequality for this distance, under the mild assumption that the order of integrals and derivatives can be interchanged.

\begin{lemma}
\label{lemma:rti}
Let $Q_{\infty} \in \mathcal{G}_\infty, Q_1\in \mathcal{G}_1$ be the least-favorable distributions (as defined in Definition~\ref{def:lfd}), and let $R_1 \in \mathcal{G}_1$ be any other distribution in the post-change family. Then 
\begin{equation*}
\mathbb{D}_{\texttt{\textup{F}}}\left(Q_1 \| Q_{\infty}\right)\leq \mathbb{D}_{\texttt{\textup{F}}}\left(R_1 \| Q_{\infty}\right) - \mathbb{D}_{\texttt{\textup{F}}}\left(R_1 \| Q_1\right).
\end{equation*}
\end{lemma}
\begin{proof}
Consider a convex set of densities \begin{align*}
\bigl\{x\mapsto q_{\xi}(x): q_{\xi}(x)=\xi q_1(x)+(1-\xi) r_1(x), \xi \in [0,1]\bigr\},
\end{align*}
where $q_1$ and $r_1$ are densities of $Q_1$ and $R_1$, respectively. 
Let $Q_{\xi}$ denote the distribution characterized by density $q_{\xi}$. 
We note that $Q_{\xi} \in \mathcal{G}_1$ due to the convexity assumption on $\mathcal{G}_1$. We use $\mathcal{L}(\xi)$ to denote the Fisher divergence $\mathbb{D}_{\texttt{\textup{F}}} \left(Q_{\xi}\| Q_{\infty}\right)$:
\begin{align*}
    \mathcal{L}(\xi)&=\int \rev{\frac{1}{2}} \big\|\nabla \log q_{\xi}-\nabla\log q_{\infty}\big\|^2 q_{\xi} dx\\
    &=\int \rev{\frac{1}{2}} \big\| \nabla \log \bigl(\xi q_1+(1-\xi)r_1\bigr)-\nabla \log q_{\infty} \big\|^2  \bigl(\xi q_1+(1-\xi)r_1\bigr) dx.
\end{align*}
Clearly $\mathcal{L}(\xi)$ is minimized at $\xi=1$, and ${\partial \mathcal{L}(\xi)}/{\partial\xi}\mid_{\xi=1^-}\le 0$. 
Letting $\mathcal{L}^{\prime}(\xi)={\partial \mathcal{L}(\xi)}/{\partial\xi}$, we have 
\begin{equation*}
\mathcal{L}^{\prime}(\xi)=\int \rev{\frac{1}{2}} \bigl(q_1-r_1\bigr)\big\|\nabla \log q_{\xi}-\nabla \log q_{\infty}\big\|^2 d x
\\
+\int  q_{\xi} \nabla\left( \frac{q_1-r_1}{q_{\xi}} \right)^T \bigl(\nabla\log q_{\xi} -\nabla \log q_{\infty}\bigr)dx.
\end{equation*}
This implies 
\begin{align*}
\label{eq:diff}
\mathcal{L}^{\prime}(1^{-}) &= \int \rev{\frac{1}{2}} \bigl(q_1-r_1\bigr)\big\|\nabla \log q_1- \nabla \log q_{\infty}\big\|^2 dx\notag 
 \quad +\int q_1 \nabla\left(\frac{q_1-r_1}{q_1}\right)^T\bigl(\nabla \log q_1-\nabla \log q_{\infty}\bigr) dx\notag\\
& = \mathbb{D}_{\texttt{\textup{F}}}\left(Q_1 \| Q_{\infty}\right)-\int\underbrace{\rev{\frac{1}{2}}  r_1\big\|\nabla\log q_1-\nabla \log q_{\infty}\big\|^2}_{\text{term 1}}   +\underbrace{ q_1 \nabla\left(\frac{q_1-r_1}{q_1}\right)^T \bigl(\nabla \log q_1-\nabla\log q_{\infty}\bigr)}_{\text{term 2}}dx. \eqnum
\end{align*}
For term 1, we have 
\begin{align*}
\label{eq:term1}
\frac{1}{2}  r_1\big\| \nabla \log q_1-\nabla\log q_{\infty}\big\|^2  &=\frac{1}{2} r_1\big\|\nabla\log q_1-\nabla\log r_1\big \|^2 
+\frac{1}{2} r_1\big\|\nabla\log r_1-\nabla\log q_{\infty}\big\|^2  \\
&\quad   +\underbrace{r_1\bigl(\nabla \log q_1-\nabla\log r_1\bigr)^T \bigl(\nabla \log r_1-\nabla\log q_{\infty}\bigr)}_{\text{term 1(a)}}. \eqnum
\end{align*}
We note that,
\begin{align}
\int \rev{\frac{1}{2}} r_1
\big\|\nabla\log q_1-\nabla\log r_1\big\|^2 dx &= \mathbb{D}_{\textit{F}}(R_1\|Q_1), \label{eq: fisher1}\\
\int \rev{\frac{1}{2}} r_1
\big\|\nabla\log r_1-\nabla\log q_{\infty}\big\|^2 dx &= \mathbb{D}_{\textit{F}}(R_1\|Q_{\infty}). \label{eq: fisher2}
\end{align}
For term 2, we note that 
\begin{align*}
\nabla\left( \frac{q_1-r_1}{q_1}\right)
=\frac{r_1}{q_1}\bigl(\nabla\log q_1-\nabla\log r_1\bigr).
\end{align*}
Therefore, 
\begin{equation}
q_1 \nabla\left(\frac{q_1-r_1}{q_1}\right)^T \bigl(\nabla \log q_1-\nabla\log q_{\infty}\bigr) \\
=r_1\bigl(\nabla \log q_1-\nabla\log r_1\bigr)^T \bigl(\nabla \log q_1-\nabla\log q_{\infty}\bigr).
\label{eq:term2}
\end{equation}
Combining the last term in Equation (\ref{eq:term1}) with Equation (\ref{eq:term2}),
\begin{align*}
\label{eq: comb_term12}
-\text{term 1(a)} + \text{term 2} 
&=r_1\bigl(\nabla \log q_1-\nabla\log r_1\bigr)^T  
     \bigl(\nabla \log q_1-\nabla\log q_{\infty} -\nabla\log r_1 +\nabla\log q_{\infty} \bigr) \\
&=r_1\|\nabla \log q_1 - \nabla \log r_1\|^2. \eqnum
\end{align*}
Plugging Equations (\ref{eq: fisher1}), (\ref{eq: fisher2}), and (\ref{eq: comb_term12}) into Equation~(\ref{eq:diff}), 
\begin{align*}
&\mathcal{L}^{\prime}(1^{-})=\mathbb{D}_{\texttt{\textup{F}}}\left(Q_1 \| Q_{\infty}\right)+\mathbb{D}_{\texttt{\textup{F}}}\left(R_1 \| Q_1\right)-\mathbb{D}_{\texttt{\textup{F}}}\left(R_1 \| Q_{\infty}\right).
\end{align*}
The results follows since ${\partial \mathcal{L}(\xi)}/{\partial \xi}\mid_{\xi=1^-}\le 0$.
\end{proof}

\subsection{Consistency of the RSCUSUM Algorithm}\label{sec:consistency}
We now apply the Lemma~\ref{lemma:rti} to prove the consistency of the RSCUSUM algorithm. Recall that $P_\infty \in \mathcal{G}_\infty, P_1 \in \mathcal{G}_1$ are the true (but unknown) pre- and post-change distributions. Also, the expectations $\mathbb{E}_\infty$ and $\mathbb{E}_1$ denote the expectations when the change occurs at $\infty$ (no change) or at $1$, respectively. Thus, under $\mathbb{E}_\infty$, every random variable $X_n$ has law $P_\infty$, and under $\mathbb{E}_1$, every random variable $X_n$ has law $P_1$.

\begin{lemma}[Positive and Negative Drifts]
\label{lemma:drifts}
Consider the instantaneous score function $X\mapsto z(X)$ as defined in Equation~(\ref{eq:scusum_instantaneous}). Under Assumption~\ref{assumption:nearness},
\begin{align*}
\label{eq:expst}
\mathbb{E}_\infty\left[z(X_1)\right] &= \mathbb{E}_\infty\left[ \mathcal{S}_{\texttt{\textup{H}}}(X_1, Q_{\infty})-\mathcal{S}_{\texttt{\textup{H}}}(X_1, Q_{1})\right] = \mathbb{D}_{\texttt{\textup{F}}}(P_{\infty} \| Q_\infty) - \mathbb{D}_{\texttt{\textup{F}}}(P_\infty \| Q_1)<0,\; \text{and}\\
    \mathbb{E}_{1}\left[z(X_1)\right] 
 &=\mathbb{E}_1\left[ \mathcal{S}_{\texttt{\textup{H}}}(X_1, Q_{\infty})-\mathcal{S}_{\texttt{\textup{H}}}(X_1, Q_{1})\right]  \ge \mathbb{D}_{\texttt{\textup{F}}}(Q_1 \| Q_{\infty})>0.
\end{align*}
\end{lemma}
\begin{proof}
Under some mild regularity conditions, \cite{hyvarinen2005estimation} proved that
\begin{align*}
    \mathbb{D}_{\texttt{\textup{F}}} (P \| Q) =\mathbb{E}_{X\sim P} \left[\frac{1}{2}\left \| \nabla_{x} \log p(X) \right \|_2^2 + \mathcal{S}_{\texttt{\textup{H}}}( X, Q)\right].
\end{align*}
As in Theorem~\ref{theorem:applicability}, we define
$$C_R \de \mathbb{E}_{X\sim R} \left[ \frac{1}{2} \left\| \nabla_{x} \log r(X)  \right\|_2^2 \right]$$
for any distribution $R$. Then 
\begin{align*}
     \mathbb{E}_\infty[\mathcal{S}_{\texttt{\textup{H}}}(X_1,  Q_{\infty})-\mathcal{S}_{\texttt{\textup{H}}}(X_1, Q_1)]  &=\mathbb{D}_{\texttt{\textup{F}}} (P_{\infty} \| Q_{\infty})-C_{P_{\infty}}-\mathbb{D}_{\texttt{\textup{F}}} (P_{\infty} \| Q_1)+C_{P_{\infty}} =\mathbb{D}_{\texttt{\textup{F}}} (P_{\infty} \| Q_\infty)-\mathbb{D}_{\texttt{\textup{F}}} (P_{\infty} \| Q_1), 
\end{align*}
which is negative by Assumption~\ref{assumption:nearness}.  Next:
\begin{align*}
     \mathbb{E}_{1}[\mathcal{S}_{\texttt{\textup{H}}}(X_1, Q_{\infty}) -\mathcal{S}_{\texttt{\textup{H}}}(X_1, Q_1)] 
     &=\mathbb{D}_{\texttt{\textup{F}}} (P_1 \| Q_{\infty})-C_{P_{1}}-\mathbb{D}_{\texttt{\textup{F}}} (P_1 \| Q_1)+C_{P_{1}} \ge \mathbb{D}_{\texttt{\textup{F}}} (Q_1 \| Q_{\infty}),
\end{align*}
where we applied Lemma \ref{lemma:rti} with $P_1$ playing the role of $R_1$.

\end{proof}

Lemma~\ref{lemma:drifts} shows that, prior to the change, the expected mean of instantaneous RSCUSUM score  $z(X)$ is negative. Consequently, the accumulated score has a negative drift at each time $n$ prior to the change. Thus, the RSCUSUM detection score $Z(n)$ is pushed toward zero before the change point. This intuitively makes a false alarm unlikely. In contrast, after the change, the instantaneous score has a positive mean, and the accumulated score has a positive drift. Thus, the RSCUSUM detection score will increase toward infinity and lead to a change detection event. Thus, the RSCUSUM algorithm can consistently detect the change and avoid false alarms, for every possible pre- and post-change distribution pair $(P_\infty, P_1)$.

\subsection{False Alarm Analysis of the RSCUSUM Algorithm}\label{sec:MFA}
In this section, we provide a bound on the mean time to false alarm for the RSCUSUM algorithm. For the analysis, we need a parameter $\rho > 0$ that satisfies the following key condition:
\begin{equation}
\label{eq: condition}    
\mathbb{E}_\infty[\exp(z_{\rho}(X_1))]\leq 1,
\end{equation}
where $z_\rho(x)$ is defined as a scalar multiple of the instantaneous score $z(x)$ defined in \eqref{eq:scusum_instantaneous}: 
\begin{align}
\label{eq:modified_det}
z_{\rho}(x)\de \rho z(x)=\rho\bigr(\mathcal{S}_{\texttt{\textup{H}}}(x, Q_{\infty})-\mathcal{S}_{\texttt{\textup{H}}}(x, Q_{1})\bigr).
\end{align}
We emphasize that $\rho$ is a quantity that depends upon $P_\infty, Q_\infty$, and $Q_1$ but which does not depend upon $\omega$.  For any choice of $\mathcal{G}_\infty, \mathcal{G}_1, P_\infty$ that satisfies the assumptions of this paper, we can show the existence of a $\rho$ that is a solution to \eqref{eq: condition}.

\begin{lemma}[Existence of appropriate $\rho$]
    \label{lemma: lambda}
Under Assumption~\ref{assumption:nearness}, there exists $\rho>0$ such that Inequality~(\ref{eq: condition}) holds. Moreover, either 1) there exists $\rho^{\star} \in (0,\infty)$ such that the equality of~(\ref{eq: condition}) holds, or 2) for all $\rho>0$, the inequality of~(\ref{eq: condition}) is strict. As noted in \cite{wu_IT_2024}, the second case is of no practical interest.
\end{lemma}
\begin{proof}
    We give proof in the appendix.
\end{proof}

The following two theorems characterize the relationship between detection threshold and delay and mean time to false alarm.
\begin{theorem}
\label{thm:arl}
Consider the stopping rule $T_{\texttt{\textup{RSCUSUM}}}$ defined in Equation~(\ref{eq:SCUSUM_rule}). Under Assumption~\ref{assumption:nearness}, for any $\omega>0$,
    \begin{equation*}
        \mathbb{E}_\infty[T_{\texttt{\textup{RSCUSUM}}}]\geq  e^{\rho\omega}.
    \end{equation*}
    for some $\rho$ which satisfies the inequality of Equation~\eqref{eq: condition}.  To satisfy the constraint of $\mathbb{E}_\infty[T_{\texttt{\textup{RSCUSUM}}}] \geq \gamma$, it is enough to set the threshold $\omega= {(\log \gamma)}/{\rho}$.
\end{theorem}
\begin{proof}
    We give proof in the appendix.
\end{proof}

Theorem~\ref{thm:arl} implies that the ARL increases at least exponentially as the stopping threshold $\omega$ increases. 
\resp{Since the bound is valid for all those $\rho$ values for which $\mathbb{E}_\infty[\exp (z_{\rho}(X_1))] \leq 1$, the bound is valid for the biggest such number, which is $\rho^*$. The existence of such a $\rho^*$ is guaranteed by Lemma VI.6. Thus, we have
    \begin{equation}
\mathbb{E}_\infty[T_{\texttt{\textup{RSCUSUM}}}]\geq  e^{\rho^* \omega}.
    \end{equation}
As a result, by setting $\omega = {(\log \gamma)}/{\rho^*}$, we get 
\begin{equation} \mathbb{E}_\infty[T_{\texttt{\textup{RSCUSUM}}}]\geq  \gamma.
\end{equation}
We note again that $\rho^*$ depends on $P_\infty, Q_\infty$, and $Q_1$.}

\subsection{Delay Analysis of the RSCUSUM Algorithm}\label{sec:DelayAnalysis}
The following theorem gives the asymptotic performance of the RSCUSUM algorithm in terms of the detection delay.

\begin{theorem}
\label{thm:cond_edd}
   The stopping rule $T_{\texttt{\textup{RSCUSUM}}}$ satisfies
\begin{align*}
    \mathcal{L}_{\texttt{\textup{WADD}}}&(T_{\texttt{\textup{RSCUSUM}}}) \sim \mathcal{L}_{\texttt{\textup{CADD}}}(T_{\texttt{\textup{RSCUSUM}}}) \sim \mathbb{E}_{1}[T_{\texttt{\textup{RSCUSUM}}}]
    \sim 
    \frac{\omega}{ \mathbb{D}_{\texttt{\textup{F}}}(P_1\|Q_{\infty})-\mathbb{D}_{\texttt{\textup{F}}}(P_1\|Q_{1})} 
    \lesssim \frac{\omega}{\mathbb{D}_{\texttt{\textup{F}}}(Q_1 \| Q_\infty)},
\end{align*} 
as $\omega \to \infty$.

Furthermore, if we let $\omega = {(\log \gamma)}/{\rho}$ as in Theorem~\ref{thm:arl}, then we have $\mathbb{E}_\infty[T_{\texttt{\textup{RSCUSUM}}}]\geq  \gamma$, and  

\begin{equation*}
   \mathcal{L}_{\texttt{\textup{WADD}}}(T_{\texttt{\textup{RSCUSUM}}}) \sim \mathcal{L}_{\texttt{\textup{CADD}}}(T_{\texttt{\textup{RSCUSUM}}}) 
   \lesssim \frac{\log \gamma}{\rho \mathbb{D}_{\texttt{\textup{F}}}(Q_1 \| Q_\infty)},
\end{equation*}
as $\gamma \to \infty$.
\end{theorem}
\begin{proof}
    We give proof in the appendix. 
\end{proof}

In the above theorem, we have used the notation $g(c)\lesssim h(c)$ as $c\to c_0$ to indicate that $\lim \sup {g(c)}/{h(c)} \leq 1$ as $c\to c_0$ for any two functions $c\mapsto g(c)$ and $c\mapsto h(c)$.  

Theorem~\ref{thm:cond_edd} implies that the {expected detection delay} (EDD), $\mathbb{E}_{1}[T_{\texttt{\textup{RSCUSUM}}}]$, increases at most linearly as the stop threshold $\omega$ increases for large values of $\omega$. Again, since the bound is valid for all those $\rho$ values for which $\mathbb{E}_\infty[\exp (z_{\rho}(X_1))] \leq 1$, the bound is valid for the biggest such number, which is $\rho^*$. This gives
\begin{equation}
   \mathcal{L}_{\texttt{\textup{WADD}}}(T_{\texttt{\textup{RSCUSUM}}}) \sim \mathcal{L}_{\texttt{\textup{CADD}}}(T_{\texttt{\textup{RSCUSUM}}}) 
   \lesssim \frac{\log \gamma}{\rho^* \mathbb{D}_{\texttt{\textup{F}}}(Q_1 \| Q_\infty)},\;\;
\end{equation}
as $\gamma \to \infty$.

\section{Identification of the Least Favorable Distributions}
\label{sec:least_favorable_distribution}

In this section, we revisit the construction considered in Theorem~\ref{theorem:applicability}.  Consider a general parametric distribution family $\mathcal{P}$ defined on $\mathcal{X}$. We use $\mathcal{P}_m$ to denote a set of a finite number of distributions belonging to $\mathcal{P}$, namely \begin{align*}
    \mathcal{P}_m = \{P_i,\; i=1,\dots, m:\; P_i\in \mathcal{P}\},\; m\in \mathbb{N}^{+}.
\end{align*}
We use $p_i$ to denote the density of each distribution $P_i, \;i=1, \dots, m$. Then, we define a convex set of densities 
\begin{equation}
\label{eq:convex_post_family}
    \mathcal{A}_m \de \biggl\{ x \mapsto \sum_{i=1}^m \alpha_i p_i(x):  \sum_{i=1}^m \alpha_i=1, \alpha_i \geq 0\biggr\}. 
\end{equation}
We further define a set of functions
\begin{equation}
\label{eq:convex_gradient_density}
    \mathcal{B}_m \de \biggl\{ x \mapsto \sum_{i=1}^m \beta_i(x) \nabla_x \log  p_i(x) \\
    \text{ s.t. } \; \sum_{i=1}^m \beta_i(x)=1,\; \beta_i(x) \geq 0,\; p_i \in \mathcal{P}_m \biggr\}. 
\end{equation}

We provide a result to identify the distribution in the set $\mathcal{A}_m$ that is nearest to some distribution $R \not\in \mathcal{A}_m$.

\begin{theorem} \label{thm_general_LFD}
    Let $R$ be some distribution such that $R \not\in \mathcal{A}_m$.  Assume that there exists an element $P_* \in \mathcal{A}_m$  (with density $p_*$) such that
\begin{equation}
  \mathbb{E}_{P_*}\biggl[ \rev{\frac{1}{2}} \|\nabla_x \log p_*(X) -\nabla_x \log r(X) \|_2^2 \biggr]
    \\
    = \min_{p \in \mathcal{A}_m, \phi \in \mathcal{B}_m} \mathbb{E}_{P} \biggl[ \rev{\frac{1}{2}} \|\phi (X) -\nabla_x \log r(X) \|_2^2 \biggr]. 
    \label{eq10}
\end{equation}
    Then, we have
\begin{equation*}
    \mathbb{E}_{P_*}\biggl[ \rev{\frac{1}{2}} \|\nabla_x \log p_*(X) -\nabla_x \log r(X) \|_2^2 \biggr]
    \\
    = \min_{p \in \mathcal{A}_m} \mathbb{E}_{P} \biggl[ \rev{\frac{1}{2}} \|\nabla_x \log p(X) -\nabla_x \log r(X) \|_2^2 \biggr].
\end{equation*}
\end{theorem}

\begin{proof}
    For any $p \in \mathcal{A}_m$, there exist $w_i$ such that $p = \sum_{i=1}^m w_i p_i$, where $w_i \geq 0$ and $\sum_{i=1}^m w_i = 1$. Direct calculations give
    \begin{align*}
        \mathbb{E}_{P} \biggl[ \rev{\frac{1}{2}} \|\nabla_x \log p(X) -\nabla_x \log r(X) \|_2^2 \biggr] &= \mathbb{E}_{P}\biggl[ \rev{\frac{1}{2}} \biggl\| \frac{\nabla_x p(X)}{ p(X)} -\nabla_x \log r(X)\biggr\|_2^2 \biggr] 
         \\
        &= \mathbb{E}_{P}\biggl[ \rev{\frac{1}{2}} \biggl\| \frac{\sum_{i=1}^m w_i \nabla_x p_i(X)}{ \sum_{i=1}^m w_i p_i(X)} -\nabla_x \log r(X) \biggr\|_2^2 \biggr] 
        \\
        &= \mathbb{E}_{P}\biggl[ \rev{\frac{1}{2}} \biggl\| \sum_{i=1}^m u_i(X) \nabla_x \log p_i(X) - \nabla_x \log r(X) \biggr\|_2^2 \biggr],
    \end{align*}
    where 
    $$u_i(x) = \frac{w_ip_i(x)}{\sum_{k=1}^m w_k p_k(x)}$$
    for all $i=1,\ldots,m$, and $\sum_{i=1}^m u_i(x)=1$. Clearly,
    $$\nabla_x \log u_i(x) - \nabla_x \log u_j(x)  = \nabla_x \log  p_i(x) - \nabla_x \log  p_j(x)$$
    for all $1 \le i,j \le m$.
    
    Using Condition~(\ref{eq10}), the quantity above is minimized at $p = p_*$, which concludes the proof. 
\end{proof}

Next, we provide a method to find the LFD in a class of Gaussian mixture models.  \resp{We begin by defining a $V$-norm:}

\begin{definition}[$V$-norm]
    \resp{The $V$-norm of a vector $\theta$ is given as \begin{equation}\| \theta\|_V \de \sqrt{\theta^T V^{-2}\theta}
    \end{equation}}
\end{definition}

\begin{theorem}
\label{theorem:gaussdisjoint}
    Let $\mathbb{M}_\infty, \mathbb{M}_1 \subset \mathbb{R}^d$ be disjoint, convex, and compact sets in $d-$dimensional Euclidean space.  Fix a symmetric, positive-definite matrix $V \in \mathbb{R}^{d \times d}$.  For $\theta \in \mathbb{M}_\infty \cup \mathbb{M}_1$, let $G_\theta$ be the Gaussian distribution with density $g_\theta$ parameterized by covariance matrix $V$ and mean $\theta: G_\theta = \mathcal{N}(\theta, V)$.  Let the distribution class $\mathcal{G}_\infty$ be the convex hull of all Gaussian $\{G_\theta : \theta \in \mathbb{M}_\infty \}$ and let $\mathcal{G}_1$ be the corresponding convex hull of $\{G_\theta : \theta \in \mathbb{M}_1 \}$. Then, $\mathcal{G}_\infty, \mathcal{G}_1$ are disjoint.
\end{theorem}
\begin{proof}
    Suppose for contradiction that there exists distribution $R$ with density $r$ such that $R \in \mathcal{G}_\infty, R \in \mathcal{G}_1$.

As $R \in \mathcal{G}_\infty$, we can pick $\mu_\infty^1, \mu_\infty^2, ..., \mu_\infty^m$ from $\mathbb{M}_\infty$ and $w_\infty^1, w_\infty^2, ..., w_\infty^m \in \mathbb{R}$ such that $w_\infty^i \geq 0$ and $\sum_{i=1}^m w_\infty^i = 1$ where 
$$r(x) = \sum_{i=1}^m w_\infty^i g_{\mu_\infty^i}(x).$$  
Consider the score of $R$ using the result of Theorem~\ref{thm_general_LFD}:
\begin{align*}
\nabla_x \log r(x) &= -\sum_{i=1}^m \beta_\infty^i(x) V^{-1}[x - \mu_\infty^i]  =-V^{-1}x + V^{-1}\sum_{i=1}^m \beta_\infty^i(x) \mu_\infty^i, \eqnum
\end{align*}
for some function $\beta_\infty$ given by Theorem~\ref{thm_general_LFD}.  But as $R \in \mathcal{G}_1$, we can also pick $\mu_1^1, \mu_1^2,..., \mu_1^n \in \mathbb{M}_1$ and $w_1^1, w_1^2, ..., w_1^n \in \mathbb{R} $ such that $ w_1^i \geq 0$ and $\sum_{i=1}^n w_1^i = 1$ where 
$$r(x) = \sum_{i=1}^n w_1^i g_{\mu_1^i}(x).$$
Then:
\begin{align*}
\nabla_x \log r(x) &= -\sum_{i=1}^n \beta_1^i(x) V^{-1}[x - \mu_1^i] = -V^{-1}x +V^{-1}\sum_{i=1}^n \beta_1^i(x) \mu_1^i, \eqnum
\end{align*}
for some function $\beta_1$ given by Theorem~\ref{thm_general_LFD}.  Certainly, these different expressions for the score of $R$ must be equivalent:
\begin{equation}
 -V^{-1}x +V^{-1} \sum_{i=1}^n \beta_1^i(x) \mu_1^i\\
 =  -V^{-1}x +V^{-1} \sum_{i=1}^m \beta_\infty^i(x) \mu_\infty^i.
\end{equation}
Adding $V^{-1}x$ to both sides and left-multiplying by $V$, we have:
\begin{align}
\sum_{i=1}^n \beta_1^i(x) \mu_1^i  =  \sum_{i=1}^m \beta_\infty^i(x) \mu_\infty^i.
\end{align}
For any $x$, $\sum_{i=1}^n \beta_1^i(x) \mu_1^i \in \mathbb{M}_1$, and $\sum_{i=1}^m \beta_\infty^i(x) \mu_\infty^i \in \mathbb{M}_\infty$ as $\mathbb{M}_\infty, \mathbb{M}_1$ are defined to each be convex.  But we have defined $\mathbb{M}_\infty$ and $\mathbb{M}_1$ to be disjoint.  Thus, there is a contradiction and $\mathcal{G}_\infty, \mathcal{G}_1$ are disjoint.
\end{proof}

\begin{remark}
    As $\mathbb{M}_\infty, \mathbb{M}_1$ are compact and as the squared $V$-norm is continuous with respect to its arguments, we know that 
    \begin{equation}
    \label{eq:nearest_means}
        (\hat{\mu}_\infty, \hat{\mu}_1) = \argmin_{(\mu_\infty, \mu_1) \in \mathbb{M}_\infty \times \mathbb{M}_1} \| \mu_\infty - \mu_1 \|_V^2
    \end{equation}
    exists.
\end{remark}
\begin{theorem}
\label{theorem:gaussvnorm}
Let $\mathbb{M}_\infty, \mathbb{M}_1, V, \mathcal{G}_\infty, \mathcal{G}_1$ be defined as in Theorem~\ref{theorem:gaussdisjoint}.
     Let $\hat{\mu}_\infty, \hat{\mu}_1$ be defined as in \eqref{eq:nearest_means}.  
    Then, the distributions $G_{\hat{\mu}_\infty}, G_{\hat{\mu}_1}$ are the nearest elements of $\mathcal{G}_\infty, \mathcal{G}_1$ under the Fisher distance: 
    \begin{equation}
    \label{eq:min_equality}
        \mathbb{D}_{\texttt{\textup{F}}}( G_{\hat{\mu}_1} \| G_{\hat{\mu}_\infty}) = \resp{\min_{(R_\infty, R_1) \in \mathcal{G}_\infty \times \mathcal{G}_1} \mathbb{D}_{\texttt{\textup{F}}}(R_1 \| R_\infty)}.
    \end{equation}
\end{theorem}

\begin{proof}

We begin by showing that
\begin{equation}
\label{eq:forward_ineq}
    \mathbb{D}_{\texttt{\textup{F}}}( G_{\hat{\mu}_1} \| G_{\hat{\mu}_\infty}) \geq \resp{\min_{(R_\infty, R_1) \in \mathcal{G}_\infty \times \mathcal{G}_1} \mathbb{D}_{\texttt{\textup{F}}}(R_1 \| R_\infty)}.
\end{equation}
Expanding the right-hand-side of \eqref{eq:forward_ineq}:
\begin{align*}
\resp{\min_{(R_\infty, R_1) \in \mathcal{G}_\infty \times \mathcal{G}_1}}\,  &\mathbb{E}_{R_1} \biggl[ \rev{\frac{1}{2}} \bigg\| \nabla_x \log r_1(X)  -\nabla_x \log r_\infty(X) \bigg\|_2^2 \biggr] \\
&\le\quad \mathbb{E}_{G_{\hat{\mu}_1}} \biggl[ \rev{\frac{1}{2}} \bigg\| \nabla_x \log g_{\hat{\mu}_1}(X) -\nabla_x \log g_{\hat{\mu}_\infty}(X) \bigg\|_2^2 \biggr] = \mathbb{D}_{\texttt{\textup{F}}}(G_{\hat{\mu}_1} || G_{\hat{\mu}_\infty}). \eqnum
\end{align*} 

We will now prove \eqref{eq:min_equality} by proving the converse of \eqref{eq:forward_ineq}. Consider an arbitrary element of $\mathcal{G}_\infty$. From the construction of $\mathcal{G}_\infty$, this element can be written as 
$$r_\infty(x) = \sum_{i=1}^m w_\infty^i g_{\mu_\infty^i}(x)$$
for some $\mu_\infty^i, ..., \mu_\infty^m \in \mathbb{M}_\infty$ and for some $w_\infty^1, ..., w_\infty^m$ such that $w_\infty^i \geq 0$ and such that $\sum_{i=1}^m w_\infty^i = 1$.  We similarly consider an arbitrary element of $\mathcal{G}_1$ and observe that it can be written as 
$$r_1(x) = \sum_{i=1}^n w_1^i g_{\mu_1^i}(x)$$
for some $\mu_1^1, ..., \mu_1^n \in \mathbb{M}_1$ and for some $w_1^1, ..., w_1^n$ such that $w_1^i \geq 0$ and $\sum_{i=1}^n w_1^i = 1$.

Next, we express $\mathbb{D}_{\texttt{\textup{F}}}(R_1 \| R_\infty)$:
\begin{equation}
\label{eq:reduction_1}
\mathbb{D}_{\texttt{\textup{F}}}(R_1 \| R_\infty) \\
= \mathbb{E}_{R_1} \bigg[ \rev{\frac{1}{2}} \bigg\| \nabla_x \log r_1(X) - \nabla_x \log r_\infty (X)  \bigg\|^2_2\bigg].
\end{equation}
We substitute $\nabla_x \log r_1(x) = -\sum_{i=1}^n \beta_1^i(x)V^{-1}(x-\mu_1^i)$ and $\nabla_x \log r_\infty(x) = -\sum_{i=1}^m \beta_\infty^i(x)V^{-1}(x-\mu_\infty^i)$, where $\beta_\infty, \beta_1$ are functions given by Theorem~\ref{thm_general_LFD}.  Then,  \eqref{eq:reduction_1} reduces to:
\begin{align}
\mathbb{E}_{R_1} \bigg[ \rev{\frac{1}{2}}  \bigg\| \sum_{i=1}^m \beta_\infty^i(X)  \mu_\infty^i -\sum_{i=1}^n \beta_1^i(X) \mu_1^i\bigg\|^2_V\bigg],
\end{align}
But by the definition of the convex hull, we know that for any $x$, $\sum_{i=1}^m \beta_\infty^i(x)\mu_\infty^i \in \mathbb{M}_\infty$ and $\sum_{i=1}^n \beta_1^i(x)\mu_1^i \in \mathbb{M}_1$.  Clearly, for all $x$:
\begin{align}
\bigg\| \sum_{i=1}^m \beta_\infty^i(x)  \mu_\infty^i -\sum_{i=1}^n \beta_1^i(x) \mu_1^i\bigg\|^2_V \geq \| \hat{\mu}_{\infty} - \hat{\mu}_{1}\|^2_V,
\end{align}
and therefore
\begin{equation*}
\begin{split}
\mathbb{E}_{R_1}&\bigg[ \rev{\frac{1}{2}} \bigg\| \sum_{i=1}^m \beta_\infty^i(X)  \mu_\infty^i -\sum_{i=1}^n \beta_1^i(X) \mu_1^i\bigg\|^2_V \bigg]
 \geq \mathbb{E}_{G_{\hat{\mu}_1}}\bigg[ \rev{\frac{1}{2}} \| \hat{\mu}_{\infty} - \hat{\mu}_{1} \|^2_V \bigg] 
= \mathbb{D}_{\texttt{\textup{F}}}(G_{\hat{\mu}_{1}} \| G_{\hat{\mu}_{\infty}}).
\end{split}
\end{equation*}

\end{proof}

\begin{remark}
Although the families $\mathcal{G}_\infty, \mathcal{G}_1$ defined in this section contain both Gaussian distributions and Gaussian mixture models, the nearest pair of distributions is always a pair of Gaussian distributions.
\end{remark}

\section{Numerical Simulations}
\label{sec:results}

In this section, we present numerical results for synthetic data to demonstrate that the RSCUSUM algorithm can consistently detect a change in the distribution of the data stream.  We will further compare the performance of the RSCUSUM algorithm against the performance of a Nonrobust-SCUSUM algorithm, where the Hyv\"arinen scores are calculated using arbitrary distributions $H_\infty \in \mathcal{G}_\infty, H_1 \in \mathcal{G}_1$:

\begin{algorithm}
\DontPrintSemicolon
\caption{Nonrobust-SCUSUM Quickest Change Detection Algorithm}
\label{algm:nonrobust_scusum}
\KwInput{Hyv\"arinen score functions 
$\mathcal{S}_{\texttt{\textup{H}}}(\cdot, H_{\infty})$ and $\mathcal{S}_{\texttt{\textup{H}}}(\cdot, H_{1})$ of arbitrary distributions in 
$\mathcal{G}_\infty, \mathcal{G}_1$, respectively.} 
\SetKwProg{Fn}{Initialization}{:}{}
  \Fn{}{
       Current time $k=0$,  $\omega>0$, and $\Psi(0)=0$}
\While{$\Psi(k)<\omega$}{
$k = k+1$\\
Update $\psi(X_k) = \mathcal{S}_{\texttt{\textup{H}}}(X_{k}, H_{\infty})-\mathcal{S}_{\texttt{\textup{H}}}(X_{k}, H_{1})$\\
Update $\Psi(k) = \max(\Psi(k-1)+\psi(X_k), 0)$\;
}
Record the current time $k$ as the stopping time $T_{\texttt{N-SCUSUM}}$\;
\KwOutput{$T_{\texttt{N-SCUSUM}}$}
\end{algorithm}

We will consider Gaussian mixture models following the setup of Theorem~\ref{theorem:gaussdisjoint} and a Gauss-Bernoulli Restricted Boltzmann Machine. We recall that we use EDD for a stopping time $T$ to denote its expected detected delay, defined as
$$
\textup{EDD}(T) = \mathbb{E}_1[T]. 
$$
For algorithms of the CUSUM type, it is well-known that in the i.i.d. setting, EDD coincides with WADD and CADD up to a constant (which is at most $1$, depending on the way WADD is defined) \cite{veeravalli2014quickest}.

\subsection{Gaussian Mixture Model Numerical Simulation}

We define $\mathbb{M}_\infty \subset \mathbb{R}^2$ to be the set of all points in the convex hull of $\{ (-0.25, -0.25)^T,(-1.5, -1.5)^T \}$ and further define $\mathbb{M}_1 \subset \mathbb{R}^2$ to be the set of all points in the convex hull of $\{ (0.25, 0.25)^T, (0.75, 0.75)^T\}$.  For any $\mu \in \mathbb{M}_\infty \cup \mathbb{M}_1$, we define $G_\mu = \mathcal{N}(\mu, V)$ where
\begin{equation}
V = 
    \begin{bmatrix}
        2 & 0.2 \\
        0.2 & 2
    \end{bmatrix}.
\end{equation}

We define $\mathcal{G}_\infty$ to be the convex hull of $\{G_\mu : \mu \in \mathbb{M}_\infty \}$ and define $\mathcal{G}_1$ to be the convex hull of $\{ G_\mu : \mu \in \mathbb{M}_1 \}$.
Note that $\mathcal{G}_\infty, \mathcal{G}_1$ contain both Gaussian distributions and Gaussian mixture models.  For ease of notation, we identify specific Gaussian distributions from $\mathcal{G}_\infty, \mathcal{G}_1$ in Table~\ref{table:distributions}.

\renewcommand{\arraystretch}{1.3}

\begin{table}[H]
\begin{center}
    \begin{tabular}{|c|c|c|}
    \hline
    Gaussian Distribution & Mean & Covariance \\
    \hline
    $R_1^A$ & $(0.25, 0.25)^T$ & $V$\\
    \hline
    $R_1^B$ & $(0.75, 0.75)^T$ & $V$\\
    \hline
    $R_\infty^A$ & $(-0.25, -0.25)^T$ & $V$\\
    \hline
    $R_\infty^B$ & $(-1.5, -1.5)^T$ & $V$\\
    \hline
    \end{tabular}
    \caption{Distributions used in Numerical Simulations}
    \label{table:distributions}
\end{center}
\end{table}

By Theorem~\ref{theorem:gaussdisjoint}, we know that $\mathcal{G}_\infty$ and $\mathcal{G}_1$ are disjoint.  Furthermore, by Theorem~\ref{theorem:gaussvnorm}, we can say that $R_\infty^A, R_1^A$ are the least favorable distributions over $\mathcal{G}_\infty, \mathcal{G}_1$ as it can be shown that their parameters are the nearest in $\mathbb{M}_\infty, \mathbb{M}_1$ under the $V$-norm.

Next, we demonstrate that Assumption~\ref{assumption:nearness} holds for any $P_\infty \in \mathcal{G}_\infty$.  Consider arbitrary Gaussian distributions with a common covariance matrix $V$.  If $C = \mathcal{N}(\mu_C, V)$ and $D = \mathcal{N}(\mu_D, V)$, then 
$$\mathbb{D}_{\texttt{\textup{F}}}(C \| D) = \rev{\frac{1}{2}}\| V^{-1}(\mu_C - \mu_D)\|^2.$$  
Let us further restrict ourselves to Gaussians whose means can be expressed as $s \mathbf{1}_2$, where $s \in \mathbb{R}$ and $\mathbf{1}_2 \in \mathbb{R}^{2}$ is a column vector of ones, and for any such Gaussian distribution $C$, let $s_C$ be the scalar constant such that $C = \mathcal{N}(s_C \mathbf{1}, V)$.  Then, we can write:

\begin{align}
    \mathbb{D}_{\texttt{\textup{F}}}(C \| D) = \rev{\frac{1}{2}} (s_C - s_D)^2 \| V^{-1} \mathbf{1}_2 \|^2 \propto (s_C - s_D)^2.
\end{align}

Clearly, $s_{R_\infty^A} = s_{Q_\infty}=  -0.25$, $s_{R_\infty^B} = -1.5$, and $s_{Q_1} = 0.25$.  Thus, we know that $\mathbb{D}_{\texttt{\textup{F}}}(R_\infty^A \| Q_\infty) < \mathbb{D}_{\texttt{\textup{F}}}(R_\infty^A \| Q_1)$ and $\mathbb{D}_{\texttt{\textup{F}}}(R_\infty^B \| Q_\infty) < \mathbb{D}_{\texttt{\textup{F}}}(R_\infty^B \| Q_1)$.  We know that $\mathcal{G}_\infty$ contains both Gaussians and Gaussian mixture models, but by Theorem~\ref{theorem:applicability},  Assumption~\ref{assumption:nearness} holds for all $P_\infty \in \mathcal{G}_\infty$.

To demonstrate the robustness of RSCUSUM, we first sample four different combinations of choices of $P_\infty \in \mathcal{G}_\infty, P_1 \in \mathcal{G}_1$ and show that the robust test consistently detects the change.  Then, we compare the performance of the robust test to the performance of a nonrobust test.

\begin{table}[H]
\begin{center}
    \begin{tabular}{|c|c|c|c|c|c|}
    \hline
    Trial & $P_\infty$ & $P_1$ & Algorithm &  \shortstack{ \strut Pre-Change\\  Drift} & \shortstack{\strut Post-Change \\ Drift} \\
    \hline 
    R-AA & $R_\infty^A$ & $R_1^A$ & Algorithm~\ref{algm:rscusum} & \rev{-0.0518} & \rev{0.0495} \\
    \hline
    R-AB & $R_\infty^A$ & $R_1^B$ & Algorithm~\ref{algm:rscusum} & \rev{-0.0511} & \rev{0.155} \\
    \hline
    R-BA & $R_\infty^B$ & $R_1^A$ & Algorithm~\ref{algm:rscusum} & \rev{-0.311} & \rev{0.0519} \\
    \hline
    R-BB & $R_\infty^B$ & $R_1^B$ & Algorithm~\ref{algm:rscusum} & \rev{-0.309} & \rev{0.157} \\
    \hline
    N & $R_\infty^A$ & $R_1^A$ & Algorithm~\ref{algm:nonrobust_scusum}  & \rev{0.124} & \rev{0.584}\\
    \hline
    \end{tabular}

\caption{The RSCUSUM Algorithm (Algorithm~\ref{algm:rscusum}) is simulated for multiple choices of $P_\infty, P_1$, and the Nonrobust-SCUSUM Algorithm (Algorithm~\ref{algm:nonrobust_scusum}) is simulated with $H_\infty = R_\infty^B, H_1 = R_1^B$.  Drifts are averaged over 50,000 sample paths. }
\label{table:gauss_drifts}
\end{center}
\end{table}

Trials R-AA, R-AB, R-BA, and R-BB are robust and use the least-favorable distributions in the calculation of their instantaneous detection scores, while trial N is a nonrobust test and uses non-least-favorable distributions (in this case, $H_\infty = R_\infty^B$ and $H_1 = R_1^B$) in the calculation of the instantaneous detection scores.  For each robust test, the pre-change drift $\mathbb{E}_\infty [z(X_1)]$ is negative and that the post-change drift $\mathbb{E}_1 [z(X_1)]$ is positive, consistent with Lemma~\ref{lemma:drifts}.  Furthermore, for the nonrobust test, the pre-change drift is positive, making a false alarm likely.

Next, we run each of the above trials to measure the mean time to false alarm and expected detection delay.
\begin{figure}[!t]
\begin{center}
\begin{subfigure}[b]{0.48\textwidth}
\includegraphics[width=8cm]{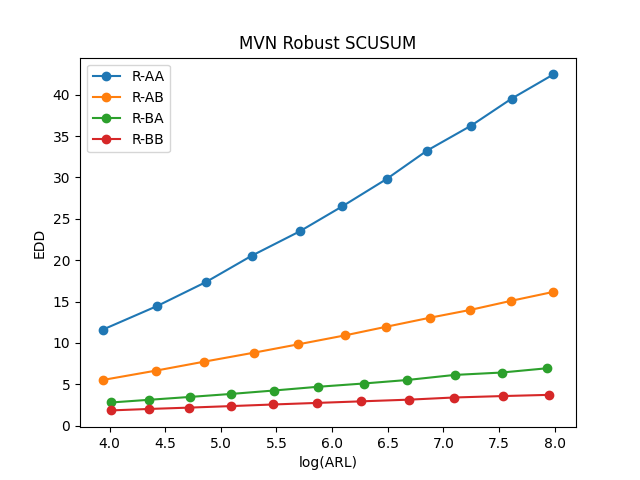}
\caption{\rev{Robust Tests}}
\end{subfigure}
\begin{subfigure}[b]{0.48\textwidth}
\includegraphics[width=8cm]{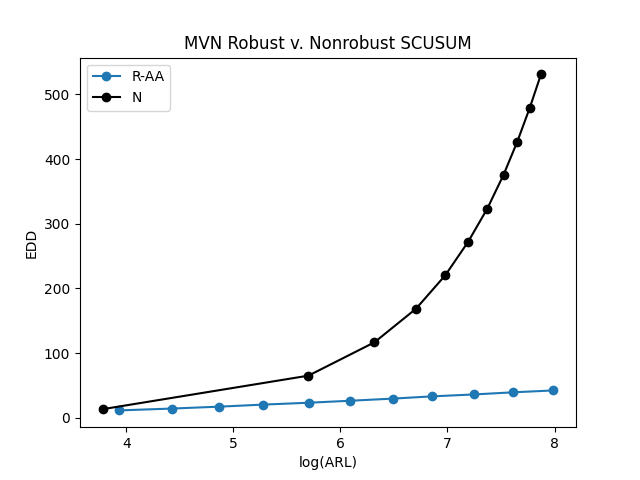}
\caption{\rev{Robust vs. Nonrobust}}
\end{subfigure}
\caption{ARL vs EDD curves for the trials of Table~\ref{table:gauss_drifts}, averaged over 10,000 sample paths.  (a) The RSCUSUM Algorithm (Algorithm~\ref{algm:rscusum}) is simulated with various choices of $P_\infty, P_1$; (b) The Nonrobust-SCUSUM algorithm (Algorithm~\ref{algm:nonrobust_scusum}) is simulated with $H_\infty = R_\infty^B, H_1 = R_1^B$ and is compared against a robust test from part (a).}
\label{fig:gauss_results}
\end{center}
\end{figure}
Figure~\ref{fig:gauss_results}(a) demonstrates that the RSCUSUM Algorithm detects the change-point for many choices of $P_\infty \in \mathcal{G}_\infty, P_1 \in \mathcal{G}_1$.  Consistent with Theorems~\ref{thm:arl} and \ref{thm:cond_edd}, the EDD increases at most linearly with respect to a bound on log-ARL when ARL becomes arbitrarily large.  The asymptotically linear relationship between ARL and EDD is consistent with that of the SCUSUM Algorithm when $P_\infty, P_1$ are known precisely.
Conversely, Figure~\ref{fig:gauss_results}(b) demonstrates that for a particular Nonrobust Algorithm applied to these uncertainty sets $\mathcal{G}_\infty, \mathcal{G}_1$, the EDD increases exponentially with respect to a bound on log-ARL.

\begin{figure*}[!t]
\centering
\begin{subfigure}[b]{0.32\textwidth}
\centering
\includegraphics[width=\textwidth]{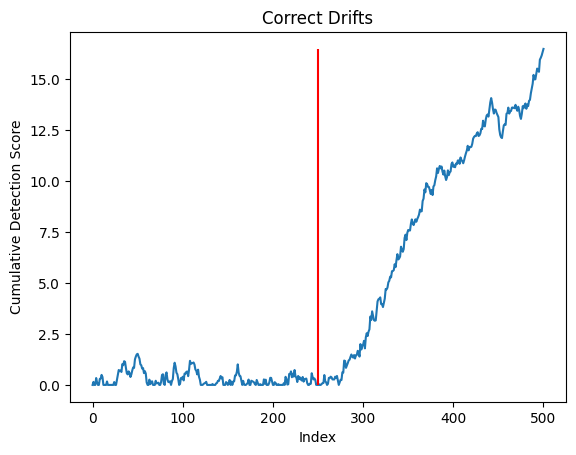}
\caption{\rev{Correct Drifts}}
\end{subfigure}
\begin{subfigure}[b]{0.32\textwidth}
\centering
\includegraphics[width=\textwidth]{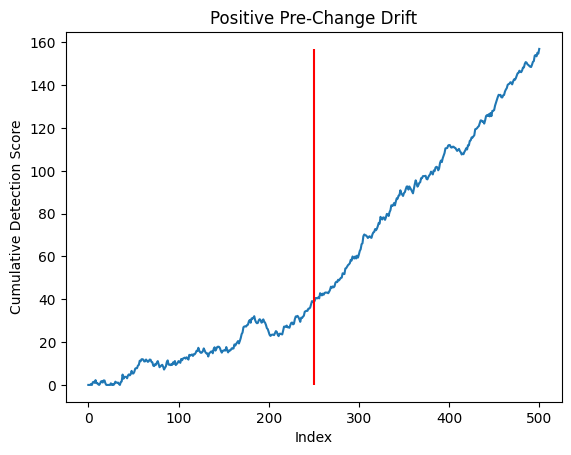}
\caption{\rev{Positive Pre-change Drift}}
\end{subfigure}
\begin{subfigure}[b]{0.32\textwidth}
\centering
\includegraphics[width=\textwidth]{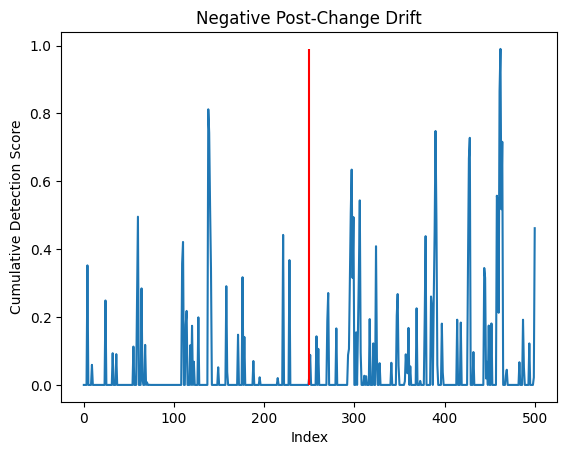}
\caption{\rev{Negative Post-change Drift}}
\end{subfigure}
\caption{Select cumulative detection score paths.  In each sample path, a change-point occurs at index 250, illustrated by a vertical red line. (a) A sample path from trial R-AA. (b) A sample path from trial N (Algorithm~\ref{algm:nonrobust_scusum} with $H_\infty = R_\infty^B, H_1 = R_1^B$).  (c) A sample path from a Nonrobust algorithm with $H_\infty = \mathcal{N}((-0.75,-0.75)^T, V)$, $H_1 = \mathcal{N}((1.5,1.5)^T, V)$, $P_\infty = R_\infty^A, P_1 = R_1^A$.  Note that this sample path is not from the experiment of Table~\ref{table:gauss_drifts}.}
\label{fig:paths}
\end{figure*}

As shown in Lemma~\ref{lemma:drifts}, the RSCUSUM Algorithm features negative pre-change drifts and positive post-change drifts.  A sample path with these drifts is illustrated in Figure~\ref{fig:paths}(a). 
The negative pre-change drift is essential to prevent false alarms; when the pre-change drift is positive (as is the case for the sample path in Figure~\ref{fig:paths}(b) above), a very large detection threshold $\omega$ must be set in order to achieve a lengthy mean time to false alarm.  Conversely, the positive post-change drifts reduce the detection delay; a negative post-change drift makes the detection delay very lengthy (as is the case for Figure~\ref{fig:paths}(c)).

\subsection{Gauss-Bernoulli Restricted Boltzmann Machine}

The Gauss-Bernoulli Restricted Boltzmann Machine (GBRBM) \cite{gbrbm_without_tears} is a model class that defines a probability density function over $\mathbb{R}^v$ 
and utilizes an $h$-dimensional latent vector.  Parameterized by 
$W \in \mathbb{R}^{v \times h}, b \in \mathbb{R}^v$, 
and $c \in \mathbb{R}^h$, the probability density function can be expressed as:
\begin{equation}
    p(x) = \frac{1}{Z}\exp(-E(x))
\end{equation}
where 
\begin{equation}
    E(x) = \frac{1}{2} \| x- b\|^2 - \bigg(\text{Softplus} (W^T x + c)\bigg) \mathbf{1}_h.
\end{equation}
Here, $\text{Softplus}(x) = \log(1 + \exp(x))$ and $\mathbf{1}_h \in \mathbb{R}^h$ is a column vector of ones.  Note that in this simulation, we assume the $\sigma$ parameter of \cite{gbrbm_without_tears} to be equal to one. 
 Further: 
\begin{equation}
    Z = \int \exp(-E(x)) dx.
\end{equation}
Calculation of $Z$ is intractable for large $v, h$, so we instead leverage the unnormalized probability density function given by 
\begin{equation}
    \tilde{p}(x) = \exp(-E(x)).
\end{equation}
The score function $\nabla_x \log p(x)$ can be expressed in closed-form in terms of $W, b, c$:
\begin{equation}
\nabla_x \log \tilde{p}(x) = b-x + W\text{Sigmoid}(W^T x + c)
\end{equation}
where $\text{Sigmoid}(x) = \frac{1}{1+\exp(-x)}$.

We next demonstrate the robustness of RSCUSUM with Gauss-Bernoulli Restricted Boltzmann Machines (GBRBMs).  For these GBRBMs, a single weight matrix $W^* \in \mathbb{R}^{10\times 8}$ and two bias vectors $v^* \in \mathbb{R}^{10}, h^* \in \mathbb{R}^{8}$ are generated by setting their elements equal to draws from a scalar standard normal distribution.  Several scalar adjustments are added to the weight matrix element-wise in order to construct several distinct GBRBMs.  

\begin{table}[H]
\begin{center}
    \begin{tabular}{|c|c|c|c|}
    \hline
    Distribution & Weight Matrix & Visible Bias & Hidden Bias \\
    \hline
    $G_\infty^0$ & $W^* - 0.2$ & $v^*$ & $h^*$ \\
    \hline
    $G_\infty^1$ & $W^* - 0.05$ & $v^*$ & $h^*$ \\
    \hline
    $G_1^0$ & $W^* $ & $v^*$ & $h^*$ \\
    \hline
    $G_1^1$ & $W^* + 0.05$ & $v^*$ & $h^*$ \\
    \hline
    \end{tabular}
    \caption{Parameters of Distributions used in Numerical Simulations.  Addition between matrices and scalars is element-wise.}
    \label{table:rbms}
\end{center}

\end{table}

These distinct GBRBMs define a finite basis for our two uncertainty classes.  Specifically, we define  $\mathcal{G}_\infty$ to be the convex hull of $\{ G_\infty^0, G_\infty^1\}$ and further define $\mathcal{G}_1$ to be the convex hull of $\{G_1^0, G_1^1 \}$:

\begin{equation}
    \mathcal{G}_\infty = \bigg\{\sum_{i=1}^m \alpha_i G_\infty^i \,: \sum_{i=0}^1 \alpha_i = 1, \alpha_i \geq 0 \bigg\},
\end{equation}
\begin{equation}
    \mathcal{G}_1 = \bigg\{\sum_{i=1}^m \alpha_i G_1^i \,: \sum_{i=0}^1 \alpha_i = 1, \alpha_i \geq 0 \bigg\}.
\end{equation}
As each of the GBRBMs of Table~\ref{table:rbms} are linearly independent functions of $x$ when the weight matrices are distinct, the sets $\mathcal{G}_\infty, \mathcal{G}_1$ are disjoint.

Here the pre-and post-change uncertainty classes $\mathcal{G}_\infty, \mathcal{G}_1$ are constructed from finite bases (see Equation~(\ref{eq:convex_post_family})). 
Using Theorem~\ref{thm_general_LFD} and following the notation of Equation~\ref{eq:convex_gradient_density}, we endeavor to learn the functions  $\beta_\infty(x), \beta_1(x)$ that correspond to the least-favorable distributions among $\mathcal{G}_\infty, \mathcal{G}_1$:

\begin{equation}
\label{eq:beta_inf}
    \nabla_x \log q_\infty (x) =  \sum_{i=0}^1 \beta_\infty^i (x) \nabla_x \log  g_\infty^i(x), 
\end{equation}
where $\sum_{i=1}^m \beta_\infty^i(x)=1,\; \beta_\infty^i(x) \geq 0$, and
\begin{equation}
\label{eq:beta_1}
    \nabla_x \log q_1 (x) =  \sum_{i=0}^1 \beta_1^i (x) \nabla_x \log  g_1^i(x),
\end{equation}
where $\sum_{i=1}^m \beta_1^i(x)=1,\; \beta_1^i(x) \geq 0$.  We use a neural network $\operatorname{Softmax}_j\circ f_\infty(x)$ to estimate $\beta_\infty^j(\cdot)$, specifically, 
\begin{align*}
    \beta_\infty^j(x) = \operatorname{Softmax}_j\circ f_\infty(x),
\end{align*}
where $\operatorname{Softmax}_j$ denotes the $j$-th element of the Softmax function and where $f_\infty$ is given by a multi-layer perceptron (MLP) network.  The architecture of this MLP is an input layer of dimension $v=10$, a single hidden layer of dimension $5$, and an output layer of dimension $2$.  This MLP utilizes $\operatorname{ReLU}$ activation functions in hidden layers.
The use of Softmax function ensures $\sum_{i=1}^m \beta_\infty^i (x) = 1$ and $\beta_\infty^i(x) \ge 0$ for all $0 \le i \le 1$.  We further create a separate neural network $f_1(x)$ with the same architecture as that of $f_\infty$, also with $\operatorname{Softmax}$ at the output layer, to learn the functions $\beta_1^i$.

To learn the scores of $Q_\infty, Q_1$, we train $f_\infty(\cdot), f_1(\cdot)$.  During training, a set of \resp{$N=100,000$} particles $X_1, ..., X_N$ are initially sampled from an arbitrary distribution of the post-change set $\mathcal{G}_1$ using \rev{Gibbs sampling \cite{gbrbm_without_tears}}.  The network is then trained to minimize the following loss function:

\begin{equation}
\label{eq:loss}
    \frac{1}{N} \sum_{i=1}^N \bigg\|\sum_{j=0}^1 \beta_1^j(X_i)\nabla \log g_1^j(X_i) - \sum_{k=0}^1 \beta_\infty^k (X_i) \nabla g_\infty^k(X_i) \bigg\|_2^2.
\end{equation}

After each epoch of training, the $N$ particles are repeatedly and iteratively updated via Langevin dynamics \cite{gbrbm_without_tears} (without Metropolis adjustment) $K$ times:

\begin{equation}
\label{eq:langevin_update}
    X_i \leftarrow X_i + \epsilon \sum_{j=0}^1 \beta_1^j(X_i) \nabla \log g_1^j(X_i) + \sqrt{2\epsilon} t
\end{equation}
where $t$ is a noise term sampled from $\mathcal{N}(0, I_v)$ and where $\epsilon$ is a step size constant so that $X_1, ..., X_N$ remain samples of the distribution with score $\sum_{j=0}^1 \beta_1^j(\cdot) \nabla \log g_1^j(\cdot)$ even as $\beta_1$ is updated.  In this experiment, we let $K=1000, \epsilon=0.01$.

In Table~\ref{tab:lfd_coeffs}, we report the average value $\frac{1}{N}\sum_{i=1}^N\beta_\infty^j(X_i)$ and $\frac{1}{N}\sum_{i=1}^N\beta_1^j(X_i)$ over the test samples $X_1, \cdots X_N$ generated using the Langevin update step of Equation~\ref{eq:langevin_update}.  In all cases the average value of $\beta_\infty^j(x), \beta_1^j(x)$ are very close to either $1$ or $0$. This gives strong evidence that the LFD is achieved by $Q_\infty = G_\infty^1, Q_1 = G_1^0$.

\begin{table}[H]
    \centering
    \begin{tabular}{|c|c|c|}
    \hline 
         j & $0$&$1$\\
         \hline
         $\beta_\infty^j$ & \rev{2.47e-6}& \rev{1.00} \\
         \hline
         $\beta_1^j$ & \rev{1.00}& \rev{1.13e-7} \\
        \hline
    \end{tabular}
    \caption{Empirical average values of $\beta_\infty^j(X), \beta_1^j(X)$ over $N=1000$ test samples.}
    \label{tab:lfd_coeffs}
\end{table}

To proceed with Algorithm~\ref{algm:rscusum}, we need a method of calculating the Hyv\"arinen scores $\mathcal{S}_{\texttt{\textup{H}}}(\cdot, Q_\infty), \mathcal{S}_{\texttt{\textup{H}}}(\cdot, Q_1)$.  These Hyv\"arinen scores are the sum of a gradient log-density, defined in Equations~\ref{eq:beta_inf} and \ref{eq:beta_1}, along with a Laplacian term.  To estimate the Laplacian term of a mixture of GBRBM distributions, we use Hutchinson's Trick (\cite{hutchinson, wu2022score}):
\begin{align*}
\label{eq:hutchinson}
    \Delta_x \log q_\infty(x) &= \mathbb{E}_{\epsilon \sim \mathcal{N}(0, I_v)}[\epsilon^T   \nabla_x f(x)   \epsilon] = \mathbb{E}_{\epsilon \sim \mathcal{N}(0, I_v)}[\epsilon^T   \nabla_x ( \epsilon^T f(x))] \eqnum
\end{align*}
with $f(x) = \sum_{i=1}^m \beta_\infty^i(x) \nabla_x \log g_\infty^i(x)$.  In this simulation, we average over ten vectors $\epsilon$ sampled from $\mathcal{N}(0,I_v)$ in order to estimate the expectation in Equation~\ref{eq:hutchinson}.
We estimate $\Delta_x \log q_1(x)$ in a similar way.

Unlike the Gaussian Mixture Model case, there is not a convenient analytical method to verify that Assumption~\ref{assumption:nearness} holds. Thus, we verify it numerically: we sample $50,000$ samples from two choices of $P_\infty$ and estimate Fisher divergences:
\begin{table}[H]
\centering

    \begin{tabular}{|c|c|c|}
    \hline
    $P_\infty$ & $\mathbb{D}_{\texttt{\textup{F}}}(P_\infty \| Q_\infty)$ & $\mathbb{D}_{\texttt{\textup{F}}}(P_\infty \| Q_1)$ \\
    \hline
    $G_\infty^0$ & \rev{4.35} & \rev{6.93} \\
    \hline 
    $G_\infty^1$ & \rev{1.93e-10} & \rev{0.336} \\
    \hline
    \end{tabular}

\caption{Numerical estimation of the Fisher distances involved in Assumption~\ref{assumption:nearness}.  For each $P_\infty$ chosen, Assumption~\ref{assumption:nearness} holds.  Fisher Divergences are estimated by averaging over 50,000 samples.}
\label{table:condition_rbm}
\end{table}

We can see that $\mathbb{D}_{\texttt{\textup{F}}}(G_\infty^1 \| Q_\infty)$ is very close to zero, and this gives further evidence that the LFD $Q_\infty = G_\infty^1$. By Theorem~\ref{theorem:applicability}, Assumption~\ref{assumption:nearness} holds for all $P_\infty \in \mathcal{G}_\infty$.

As with the Gaussian Mixture Model case, we first sample four different combinations of choices of $P_\infty \in \mathcal{G}_\infty, P_1 \in \mathcal{G}_1$ and show that the robust test consistently detects the change.  Then, we compare the performance of the robust test to the performance of a nonrobust test.

\begin{table}[ht]
\begin{center}
    \begin{tabular}{|c|c|c|c|c|c|}
    \hline
    Trial & $P_\infty$ & $P_1$ & Algorithm & \shortstack{\strut Pre-Change\\ Drift} & \shortstack{\strut Post-Change\\ Drift} \\
    \hline 
    R-00 & $G_\infty^0$ & $G_1^0$ & Algorithm~\ref{algm:rscusum}& \resp{-2.59} & \resp{0.362} \\
    \hline
    R-01 & $G_\infty^0$ & $G_1^1$ & Algorithm~\ref{algm:rscusum}& \resp{-2.56} & \resp{1.19} \\
    \hline
    R-10 & $G_\infty^1$ & $G_1^0$ & Algorithm~\ref{algm:rscusum}& \resp{-0.324} & \resp{0.357} \\
    \hline
    R-11 & $G_\infty^1$ & $G_1^1$ & Algorithm~\ref{algm:rscusum}&\resp{-0.337} & \resp{1.18} \\
    \hline
    N & $G_\infty^1$ & $G_1^0$ & Algorithm~\ref{algm:nonrobust_scusum}& \rev{2.10} & \rev{5.81}\\
    \hline
    \end{tabular}

\caption{The RSCUSUM Algorithm (Algorithm~\ref{algm:rscusum}) is simulated for multiple choices of $P_\infty, P_1$, and the Nonrobust-SCUSUM Algorithm (Algorithm~\ref{algm:nonrobust_scusum}) is simulated with $H_\infty = G_\infty^0, H_1 = G_1^1$.  Drifts are averaged over 50,000 sample paths. }
\label{table:rbm_drifts}
\end{center}
\end{table}

Trials R-00, R-01, R-10, and R-11 are robust and use the least-favorable distributions in the calculation of their instantaneous detection scores, while trial N is a nonrobust test and uses non-least-favorable distributions in the calculation of the instantaneous detection scores.  As before, the robust drifts are negative before the change point and positive after the change point, but the nonrobust drifts are positive before the change point.

Next, we vary the detection threshold $\omega$ and generate plots comparing mean time to false alarm with the expected detection delay.
\begin{figure}[!t]
\begin{center}
\begin{subfigure}[b]{0.48\textwidth}
\includegraphics[width=8cm]{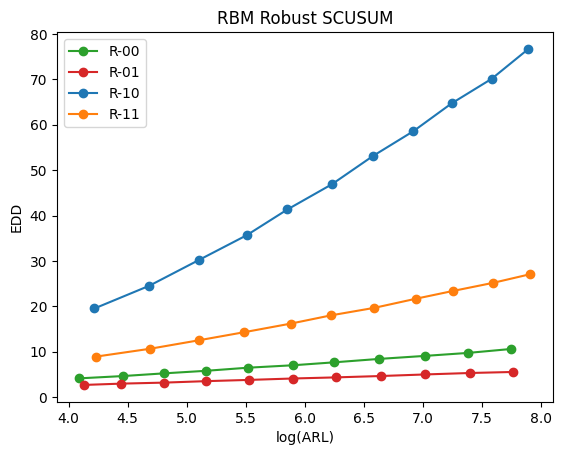}
\caption{\rev{Robust Tests}}
\end{subfigure}
\begin{subfigure}[b]{0.48\textwidth}
\includegraphics[width=8cm]{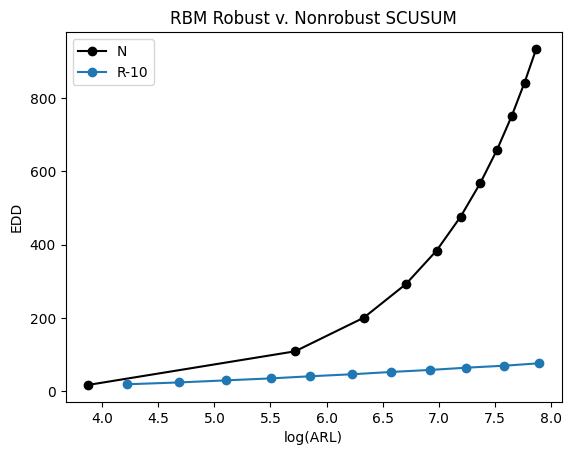}
\caption{\rev{Robust vs. Nonrobust}}
\end{subfigure}
\caption{ARL vs EDD curves for the trials of Table~\ref{table:rbm_drifts}, averaged over 10,000 sample paths.  (a) The RSCUSUM Algorithm (Algorithm~\ref{algm:rscusum}) is simulated with various choices of $P_\infty, P_1$. (b) The Nonrobust-SCUSUM algorithm (Algorithm~\ref{algm:nonrobust_scusum}) is simulated with $H_\infty = G_\infty^0, H_1 = G_1^1$ and is compared against a robust test from part (a).  }
\label{fig:rbm_results}
\end{center}
\end{figure}
Figure~\ref{fig:rbm_results}(a) demonstrates that the RSCUSUM Algorithm detects the change-point for many choices of $P_\infty \in \mathcal{G}_\infty, P_1 \in \mathcal{G}_1$.  Again, the log-ARL vs. EDD plot corroborates Theorems~\ref{thm:arl} and \ref{thm:cond_edd} as EDD increases at most linearly with log-ARL (as is the case for the SCUSUM Algorithm with precisely known $P_\infty, P_1$).  As before, a particular Nonrobust Algorithm applied to these uncertainty sets $\mathcal{G}_\infty, \mathcal{G}_1$ produces an EDD that is exponential in log-ARL (Figure~\ref{fig:rbm_results}(b)).

\section{Conclusion}
In this work, we proposed the RSCUSUM algorithm\resp{, which can detect a change in the distribution of an online high-dimensional data stream in real time.  While the SCUSUM algorithm  could accomplish this by knowing only the score function of the pre- and post-change distributions, the new RSCUSUM method can accomplish this with the additional challenge of imprecisely knowing the pre- and post-change distribution score functions.}

The RSCUSUM algorithm is defined using a novel notion of least-favorable distributions obtained using the Fisher divergence between pre- and post-change distributions. We also analyzed the average detection delay and the mean time to a false alarm of the algorithm using novel analytical techniques. 
We provided both theoretical and algorithmic methods for computing the least favorable distributions for unnormalized models\resp{, and provided an analytical solution in the case where the pre- and post- change models are specific Gaussian Mixture Models.} Numerical simulations were provided to demonstrate the performance \resp{and implementability} of our robust algorithm.

\appendices

\section{Selection of Appropriate Multiplier}

The following proof and subsequent discussion use similar arguments as the proof of Lemma 2 of \cite{wu_IT_2024}.  We note that the proofs are similar but not identical as the pre-change distribution $P_\infty$ may not be equal to the least favorable distribution $Q_\infty$.

\begin{proof}[Proof of Lemma~\ref{lemma: lambda}]
Define the function $\rho:\mapsto h(\rho)$ given by $$h(\rho)\de\mathbb{E}_\infty[\exp (z_{\rho}(X_1))]-1.$$ where $z_\rho(x)$ is defined as in \eqref{eq:modified_det}.  Observe that 
\begin{equation*}
  h^{\prime}(\rho)\de \frac{d h}{d\rho}(\rho)\\
  =\mathbb{E}_\infty[(S_{\texttt{H}}(X_1,Q_{\infty})-S_{\texttt{H}}(X_1,Q_{1}))\exp (z_{\rho}(X_1))].
\end{equation*}
Note that $h(0)=0$, and $h^{\prime}(0)=  \mathbb{D}_{\texttt{\textup{F}}}(P_\infty \| Q_\infty) - \mathbb{D}_{\texttt{\textup{F}}}(P_\infty \| Q_1) < 0$ by Assumption~\ref{assumption:nearness}. 

Next, we prove that either 1) there exists $\rho^{\star} \in (0,\infty)$ such that $h(\rho^{\star}) = 0$, or 2) for all $\rho>0$ we have $h(\rho)<0$. 

Observe that
\begin{equation*}
    h''(\rho) \de\frac{d^2 h}{d \rho}(\rho) \\
    =\mathbb{E}_\infty[(S_{\texttt{H}}(X_1,Q_{\infty})-S_{\texttt{H}}(X_1,Q_{1}))^2\exp (z_{\rho}(X_1))]\geq 0.
\end{equation*}
We claim that $h(\rho)$ is {strictly convex}, namely $h''(\rho) > 0$ for all $\rho\in [0,\infty)$. Suppose $h''(\rho) = 0$ for some $\rho \geq 0$, we must have $S_{\texttt{H}}(X_1,Q_{\infty})-S_{\texttt{H}}(X_1,Q_1) = 0$ almost surely.  This implies that 
$$\mathbb{E}_\infty[(S_{\texttt{H}}(X_1,Q_{\infty})-S_{\texttt{H}}(X_1,Q_1))]
= 0$$ 
which in turn gives 
$$\mathbb{D}_{\texttt{\textup{F}}}(P_\infty \| Q_\infty) - \mathbb{D}_{\texttt{\textup{F}}}(P_\infty \| Q_1) =0,$$
violating Assumption~\ref{assumption:nearness}. Thus, $h(\rho)$ is {strictly convex} and $h^{\prime}(\rho)$ is {strictly increasing}. 

Here, we recognize two cases: either 1) $h(\rho)$ have at most one global minimum in $(0, \infty)$, or 2) it is strictly decreasing in $[0,\infty)$. We will show that the second case is degenerate that is of no practical interest.

\textit{Case 1:} If the global minimum of $h(\rho)$ is attained at $a \in (0, \infty)$, then $h^{\prime}(a) = 0$. Since $h^{\prime}(0) < 0$ and $h(0) = 0$, the global minimum $h(a)<0$. 
Since $h^{\prime}(\rho)$ is \textit{strictly increasing}, we can choose $b > a$ and conclude that $h^{\prime}(\rho) > h^{\prime}(b) > h^{\prime}(a) = 0$ for all $\rho > b$. It follows that $\lim_{\rho \rightarrow \infty} h(\rho) = +\infty$. Combining this with the continuity of $h(\rho)$, we conclude that $h(\rho^*) = 0$ for some $\rho^* \in (0, \infty)$ and any value of $\rho \in (0, \rho^*]$ satisfies Inequality~(10).

Note that in this case, we must have 
$$P_{\infty}\left(S_{\texttt{H}}(X_1,Q_{\infty})-S_{\texttt{H}}(X_1,Q_1) \ge c \right) > 0$$
for some $c>0$. Otherwise, we have 
$$P_{\infty}\left(S_{\texttt{H}}(X_1,Q_{\infty})-S_{\texttt{H}}(X_1,Q_1) \le 0\right)=1.$$
This implies that $P_{\infty}(z_{\rho}(X_1)\le 0)=1$, or equivalently $\mathbb{E}_\infty[\exp (z_{\rho}(X_1))]< 1$ for all $\rho > 0$, and therefore leads to Case 2: $h(\rho)< 0$ for all $\rho > 0$. 
Here, 
$\mathbb{E}_\infty[\exp (z_{\rho}(X_1))]\neq 1$
since 
$$P_{\infty}(S_{\texttt{H}}(X_1,Q_{\infty})-S_{\texttt{H}}(X_1,Q_1)=0)<1;$$ 
otherwise 
$$P_{\infty}(S_{\texttt{H}}(X_1,Q_{\infty})-S_{\texttt{H}}(X_1,Q_1)=0)=1,$$
and then 
\begin{equation*}
\mathbb{E}_\infty  [S_{\texttt{H}}(X_1,Q_{\infty})-S_{\texttt{H}}(X_1,Q_1)] \\
=\mathbb{D}_{\texttt{\textup{F}}}(P_\infty \| Q_\infty) - \mathbb{D}_{\texttt{\textup{F}}}(P_\infty \| Q_1)=0,
\end{equation*}
causing the same contradiction to Assumption~\ref{assumption:nearness}.

\textit{Case 2:} If $h(\rho)$ is strictly decreasing in $(0, \infty)$, then any $\rho \in (0, \infty)$ satisfies Inequality~(10). As discussed before, in this case, we must have $$P_{\infty}\left(S_{\texttt{H}}(X_1, Q_{\infty})-S_{\texttt{H}}(X_1, Q_1) \le 0\right)=1.$$
Equivalently, all the increments of the RSCUSUM detection score are non-positive under the pre-change distribution, and $P_{\infty}(Z(n)=0)=1$ for all $n$. Accordingly, $\mathbb{E}_\infty[T_{\textit{RSCUSUM}}]=+\infty$. When there occurs change (under measure $P_1$), we also observe that RSCUSUM can get close to detecting the change-point instantaneously as $\rho$ is chosen arbitrarily large. Obviously, this case is of no practical interest.

\end{proof}

\section{Proofs of Delay and False Alarm Theorems}
\label{section:proofs}
The theoretical analysis for the mean time to a false alarm is analogous to the analysis from \cite{wu_IT_2024}. We note that the proofs are similar but not identical as the pre-change distribution $P_\infty$ may not be equal to the least favorable distribution $Q_\infty$.


\subsection{Proof of Theorem~\ref{thm:arl}}
\begin{proof}
We prove the result in multiple steps and follow the proof technique of \cite{wu_IT_2024}, which in turn uses the proof technique given in  \cite{lai1998information}. 


\textit{Part 1:}
Define 
\begin{align*}
    Z_\rho(0)=0, \quad Z_\rho(n) \de (Z_\rho(n-1)+z_\rho(X_1))^{+},\;\forall n\geq 1,
\end{align*}
where 
\begin{equation}
\label{eq:scusum_instantaneous_111}
    z_\rho(x) = \rho \left(\mathcal{S}_{\texttt{\textup{H}}}(x, Q_{\infty})-\mathcal{S}_{\texttt{\textup{H}}}(x, Q_{1})\right),
\end{equation}
where $\mathcal{S}_{\texttt{\textup{H}}}(x, Q_{\infty})$ and $\mathcal{S}_{\texttt{\textup{H}}}(x, Q_1)$ are respectively the Hyv\"arinen score functions of $Q_\infty$ and $Q_1$, the LFD pair. Also, note that $\rho$ satisfies 
\begin{equation}
\label{eq:condition_111}    
\mathbb{E}_\infty[\exp(z_{\rho}(X_1))]\leq 1,
\end{equation}
Note that 
\begin{equation}
\label{eq:modified}
T_{\texttt{\textup{RSCUSUM}}}=\inf\{n\geq 1:  Z_\rho(n)\geq \omega\rho \}.
\end{equation}
For convenience, define $\omega' \de \omega \rho$.

\textit{Part 2:} We next construct a non-negative martingale with mean $1$ under the measure $P_{\infty}$. Define a new instantaneous score function $x \mapsto \tilde{z}_{\rho}(x)$ given by 
\begin{equation*}
    \label{eq:new_instant_z_lambad}
    \tilde{z}_{\rho}(x)\de z_{\rho}(x)+\delta,
\end{equation*}
where $$\delta \de -\log \biggr(\mathbb{E}_\infty\left[\exp (z_{\rho}(X_1))\right]\biggr).$$ Further define the sequence $$\tilde{G}_n\de \exp \biggr(\sum_{k=1}^n\tilde{z}_{\rho}(X_k)\biggr),\; \forall n\geq 1.$$ 
Since $\mathbb{E}_\infty\left[\exp (z_{\rho}(X_1))\right] \leq 1$, $\delta \geq 0$. 

Suppose $X_1, X_2, \ldots$ are i.i.d according to $P_{\infty}$ (no change occurs). Then,
\begin{align*}
\mathbb{E}_\infty\left[\tilde{G}_{n+1}\mid \mathcal{F}_n\right] 
&= \tilde{G}_n\mathbb{E}_\infty[\exp(\tilde{z}_{\rho}(X_{n+1}))] =\tilde{G}_{n}e^{\delta}\mathbb{E}_\infty[\exp(z_{\rho}(X_{n+1}))]=\tilde{G}_{n},
\end{align*}
and
\begin{align*}
\label{eq:martingalemeanone}
    \mathbb{E}_\infty[\tilde{G}_n] &= \mathbb{E}_\infty\left[\exp\left(\sum_{i=1}^{n}(z_{\rho}(X_i)+\delta)\right)\right] = e^{n\delta} \prod_{i=1}^n\mathbb{E}_\infty[\exp(z_{\rho}(X_i))]=1. \eqnum
\end{align*}
Thus, under the measure $P_{\infty}$, $\{\tilde{G}_n\}_{n\geq 1}$ is a non-negative martingale with the mean $\mathbb{E}_\infty[\tilde{G}_1]=1$.

\textit{Part 3:} We next examine the new stopping rule 
\begin{equation*}
    \tilde{T}_{\texttt{\textup{RSCUSUM}}} = \inf \left\{n\geq 1: \max_{1\leq k\leq n} \sum_{i=k}^n \tilde{z}_{\rho}(X_i)\geq \omega' \right\},
\end{equation*}
where $\tilde{z}_{\rho}(X_i) = z_{\rho}(X_i)+\delta$. Since $\delta \geq 0$, we have 
$$
T_{\texttt{\textup{RSCUSUM}}} \geq \tilde{T}_{\texttt{\textup{RSCUSUM}}}, 
$$ 
and 
$\mathbb{E}_\infty[T_{\texttt{\textup{RSCUSUM}}}]\geq \mathbb{E}_\infty[\tilde{T}_{\texttt{RSCUSUM}}]$. Thus, it is sufficient to lower bound the modified stopping time $\tilde{T}_{\texttt{RSCUSUM}}$.

\textit{Part 4:} By Jensen's inequality,
\begin{equation}
\label{eq:jensen}
    \mathbb{E}_\infty[\exp(z_{\rho}(X_1))]\geq \exp\left(\mathbb{E}_\infty[z_{\rho}(X_1)]\right),
\end{equation}
with equality holds if and only if $z_{\rho}(X_1)=c$ almost surely, where $c$ is some constant. Suppose the equality of Equation~(\ref{eq:jensen}) holds, then
\begin{equation*}
     \rho (\mathbb{D}_{\texttt{\textup{F}}}(P_\infty \| Q_\infty) - \mathbb{D}_{\texttt{\textup{F}}}( P_\infty  \| Q_1) ) =\mathbb{E}_\infty[z_{\rho}(X_1)] \\
    =c 
    =\mathbb{E}_{1}[z_{\rho}(X_1)]
    \geq \rho\mathbb{D}_{\texttt{\textup{F}}}(Q_1 \| Q_\infty)
\end{equation*} 

Recalling Assumption~\ref{assumption:nearness} and dividing through by $\rho > 0$, we know that $0 > \mathbb{D}_{\texttt{\textup{F}}}(P_\infty \| Q_\infty) - \mathbb{D}_{\texttt{\textup{F}}}(P_\infty \| Q_1) \geq \mathbb{D}_{\texttt{\textup{F}}}(Q_1 \| Q_\infty) \geq 0$, the last inequality a result of the fact that the Fisher distance is nonnegative.  Thus, the inequality of Equation~(\ref{eq:jensen}) is {strict}.  Continuing from Equation~\ref{eq:martingalemeanone}:
\begin{align*}
1=&e^{n\delta} \prod_{i=1}^n \mathbb{E}_\infty [\exp(z_\rho(X_i))] \\
e^{-n\delta} =& \bigg(\mathbb{E}_\infty [\exp(z_\rho(X_1))]\bigg)^n \\
e^{-\delta} =&\mathbb{E}_\infty [\exp(z_\rho(X_1))]  > \exp(\mathbb{E}_\infty (z_\rho(X_1))),
\end{align*}
where the strictness of the final inequality follows from the strictness of the inequality in Equation~(\ref{eq:jensen}).  Taking the log, we have that $\delta < -\rho(\mathbb{D}_{\texttt{\textup{F}}}(P_\infty \| Q_\infty) - \mathbb{D}_{\texttt{\textup{F}}}(P_\infty \| Q_1))$, so 
\begin{align*}
\mathbb{E}_\infty[\tilde{z}_\rho(X_i)] &= \mathbb{E}_\infty[z_\rho(X_i)] + \delta < \rho(\mathbb{D}_{\texttt{\textup{F}}}(P_\infty \| Q_\infty) - \mathbb{D}_{\texttt{\textup{F}}}(P_\infty \| Q_1))   -\rho(\mathbb{D}_{\texttt{\textup{F}}}(P_\infty \| Q_\infty) - \mathbb{D}_{\texttt{\textup{F}}}(P_\infty \| Q_1)) = 0
\end{align*}
so we know that $\tilde{z}_\rho(X_i)$ has a negative drift under the $P_\infty$ law.  Thus, $\tilde{T}_{\texttt{RSCUSUM}}$ is not trivial.

\textit{Part 5:} Define a sequence of stopping times: 
\begin{align*}
    \eta_0 = 0, \quad \quad
    \eta_1 = \inf \left\{t:\sum_{i=1}^t \tilde{z}_{\rho}(X_i)<0\right\}, \quad \quad
    \eta_{k+1} = \inf \left\{t>\eta_k:\sum_{i=\eta_k+1}^t \tilde{z}_{\rho}(X_i)<0\right\}, \; \text{for}\;  k\geq 1.
\end{align*}
Furthermore, 
let
\begin{equation}
\label{eq:defm}
    M \de \inf \biggl\{k\geq 0: \eta_k<\infty\\
    \;\text{and} \; \sum_{i=\eta_k+1}^n\tilde{z}_{\rho}(X_i)\geq \omega' \; \text{for some}\; n>\eta_k\biggr\}.
\end{equation}
Then, it follows from the definitions that
\begin{align*}
    \tilde{T}_{\texttt{RSCUSUM}}\geq M.
\end{align*}
Hence, it is enough to lower bound $\mathbb{E}_\infty[M]$. To this end, we obtain a lower bound on the probability $P_{\infty}(M> k)$ and then utilize the identity $\mathbb{E}_\infty[M] = \sum_{k=0}^{\infty}P_{\infty}(M> k)$.

\textit{Part 6:} For the probability $P_{\infty}(M> k)$, note that since the event $\{M < k\}$ is $\mathcal{F}_{\eta_k}-$measurable, 
\begin{align*}
\label{eq:eq1}
    P_{\infty}(M> k) = \mathbb{E}_\infty[\mathbb{I}_{\{M > k\}}] 
    =\mathbb{E}_\infty[\mathbb{I}_{\{M > k\}} \mathbb{I}_{\{M \geq k\}}] 
    &=\mathbb{E}_\infty \left[\mathbb{E}_\infty[\mathbb{I}_{\{M > k\}} \mathbb{I}_{\{M \geq k\}} | \mathcal{F}_{\eta_k}]\right]\\
    &= \mathbb{E}_\infty [P_{\infty}(M\geq k+1\mid\mathcal{F}_{\eta_k})\mathbb{I}_{\{M\geq k\}}]. \eqnum
\end{align*}

Next, we bound the probability $P_{\infty}(M\geq k+1\mid\mathcal{F}_{\eta_k})$. For this, note that 
\begin{equation}
    P_{\infty}(M\geq k+1\mid\mathcal{F}_{\eta_k})= \\
    1-P_{\infty}\left(\sum_{i=\eta_k+1}^n\tilde{z}_\rho(X_i)\geq \omega'  \;\text{for some} \; n>\eta_k\mid \mathcal{F}_{\eta_k}\right). 
\end{equation}
Now, note that due to the independence of the observations, 
\begin{equation*}
P_{\infty}\left(\sum_{i=\eta_k+1}^n\tilde{z}_\rho(X_i)\geq \omega'  \;\text{for some} \; n>\eta_k\mid \mathcal{F}_{\eta_k}\right) \\
= P_{\infty}\left(\sum_{i=1}^n\tilde{z}_\rho(X_i)\geq \omega'  \;\text{for some} \; n\right).
\end{equation*}
Next, for the right-hand side of the above equation, we have
\begin{align*}
P_{\infty}\left(\sum_{i=1}^n\tilde{z}_\rho(X_i)\geq \omega'  \;\text{for some} \; n\right) 
&= \lim_{m \rightarrow \infty} P_\infty \left( \max_{n : n \leq m} \sum_{i=1}^n \tilde{z}_\rho(X_i) > \omega' \right) \\
&=\lim_{m \rightarrow \infty} P_\infty \left( \max_{n : n \leq m} \exp \left( \sum_{i=1}^n \tilde{z}_\rho(X_i)\right) > e^{\omega'} \right)\\
&=\lim_{m \rightarrow \infty} P_\infty \left( \max_{n : n \leq m} \tilde{G}_n > e^{\omega'} \right).
\end{align*}
Since $\{\tilde{G}_n\}_{n\geq 1}$ is a nonnegative martingale under $P_{\infty}$ with mean 1, by Doob's submartingale inequality~\cite{doob1953stochastic}, we have:
\begin{align}
\label{eq:doobsresult}
P_\infty \left( \max_{n : n \leq m} \tilde{G}_n > e^{\omega'} \right) \leq \frac{\mathbb{E}[\tilde{G}_m]}{e^{\omega'}} = e^{-\omega'}
\end{align}
Thus,
\begin{align*}
\label{eq:eq2_11}
    P_{\infty}(M\geq k+1\mid\mathcal{F}_{\eta_k}) 
    &= 1-P_{\infty}\left(\sum_{i=\eta_k+1}^n\tilde{z}_\rho(X_i)\geq \omega'  \;\text{for some} \; n>\eta_k\mid \mathcal{F}_{\eta_k}\right)\\
    &=1 - P_{\infty}\left(\sum_{i=1}^n\tilde{z}_\rho(X_i)\geq \omega'  \;\text{for some} \; n\right)\\
    &= 1 - \lim_{m \rightarrow \infty} P_\infty \left( \max_{n : n \leq m} \tilde{G}_n > e^{\omega'} \right) \\
    & \geq 1 - e^{-\omega'}.  \eqnum
\end{align*}

\textit{Part 7:} Finally, from \eqref{eq:eq1}, we have 
\begin{equation}
\begin{split}
\label{eq:eq1_1111}
    P_{\infty}(M> k) &= \mathbb{E}_\infty [P_{\infty}(M\geq k+1\mid\mathcal{F}_{\eta_k})\mathbb{I}_{\{M\geq k\}}]\\
    &\geq (1-e^{-\omega'})P_{\infty}(M> k-1)\\
    &\geq (1-e^{-\omega'})^2 P_{\infty}(M> k-2)\\
    &\geq \dots \geq (1-e^{-\omega'})^k
    \end{split}
\end{equation}

Using (\ref{eq:eq1_1111}), we arrive at a geometric series:
\begin{align*}
    \mathbb{E}_\infty[M] = \sum_{k=0}^{\infty}P_{\infty}(M> k)\geq \sum_{k=0}^{\infty}(1-e^{-\omega'})^{k}= e^{\omega'}.
\end{align*}
Thus, we conclude that
$$\mathbb{E}_\infty[T_{\texttt{\textup{RSCUSUM}}}]\geq \mathbb{E}_\infty[\tilde{T}_{\texttt{RSCUSUM}}]\geq \mathbb{E}_\infty[M]\geq e^{\omega'}= e^{\rho \omega}.
 $$

\end{proof}

\subsection{\resp{Proof of Theorem~\ref{thm:cond_edd}}}

\begin{proof}
Consider the random walk that is defined by 
\begin{equation*}
    Z^{\prime}(n) = \sum_{i=1}^nz(X_i), \; \text{for}\; n\geq 1.
\end{equation*}
We examine another stopping time that is given by
\begin{equation*}
     T_{\texttt{\textup{RSCUSUM}}}^{\prime} \de \inf \{n\geq 1: Z^{\prime}(n) \geq \omega\}.
\end{equation*}
Next, for any $\omega$, define $R_{\omega}$ on $\{T_{\texttt{\textup{RSCUSUM}}}^{\prime} <\infty\}$ by 
\begin{equation*}
    R_{\omega} \de Z^{\prime}(T_{\texttt{\textup{RSCUSUM}}}^{\prime}) -\omega.
\end{equation*}
$R_{\omega}$ is the excess of the random walk over a stopping threshold $\omega$ at the stopping time $T_{\texttt{\textup{RSCUSUM}}}^{\prime}$.
Suppose the change-point $\nu =1$, then $X_1, X_2,\ldots, $ are i.i.d. following the distribution $P_1$. Let $\mu$ and $\sigma^2$ respectively denote the mean $\mathbb{E}_{1}[z(X_1)]$ and the variance $\text{Var}_1[z(X_1)]$. Note that 
\begin{equation*}
    \mu =\mathbb{E}_{1}[z(X_1)]= \mathbb{D}_{\texttt{\textup{F}}}(P_1\|Q_{\infty})-\mathbb{D}_{\texttt{\textup{F}}}(P_1\|Q_{1})>0,
\end{equation*}
and \begin{align*}
    \sigma^2 &= \text{Var}_1[z(X_1)] = \mathbb{E}_{1}[z(X_1)^2]-\left(\mathbb{D}_{\texttt{\textup{F}}}(P_1\|Q_{\infty})-\mathbb{D}_{\texttt{\textup{F}}}(P_1\|Q_{1})\right)^2. 
\end{align*}
Under the mild assumptions that:
\begin{align*}
&\mathbb{E}_{1}[\mathcal{S}_{\texttt{\textup{H}}}(X_1, Q_{\infty})^2] < \infty,\; \quad \mathbb{E}_{1}[\mathcal{S}_{\texttt{\textup{H}}}(X_1, Q_1)^2] < \infty,
\end{align*}
then by \cite[Theorem 1]{lorden1970excess},
\begin{equation*}
    \sup_{\omega \geq 0}\mathbb{E}_{1}[R_{\omega}]
    \leq \frac{\mathbb{E}_{1}[(z(X_1)^{+})^2]}{\mathbb{E}_{1}[z(X_1)]}
    \leq \frac{\mathbb{E}_{1}[(z(X_1))^2]}{\mathbb{E}_{1}[z(X_1)]}
    = \frac{\mu^2+\sigma^2}{\mu},
\end{equation*}
where $z(X)^{+} = \max (z(X), 0)$.
Hence, by Wald's lemma \cite{woodroofe1982nonlinear}
\begin{equation*}
    \mathbb{E}_{1}[T^{\prime}_{\texttt{RSCUSUM}}]=\frac{\omega}{\mu}+\frac{\mathbb{E}_{1}[{R_{\omega}}]}{\mu}\leq \frac{\omega}{\mu}+\frac{\mu^2+\sigma^2}{\mu^2},\;\forall \omega \geq 0.
\end{equation*}
Observe that for any $n$, $Z^{\prime}(n)\leq Z(n)$, and therefore $T_{\texttt{\textup{RSCUSUM}}} \leq T_{\texttt{\textup{RSCUSUM}}}^{\prime}$. Thus, 
\begin{equation}
\label{eq:cadd_result}
    \mathbb{E}_{1}[T_{\texttt{\textup{RSCUSUM}}}]\leq \mathbb{E}_{1}[T_{\texttt{\textup{RSCUSUM}}}^{\prime}]\leq \frac{\omega}{\mu}+\frac{\mu^2+\sigma^2}{\mu^2},\;\forall \omega \geq 0.
\end{equation}
The contribution of the term $\frac{\mu^2 + \sigma^2}{\mu^2}$ becomes negligible as $\omega \rightarrow \infty$. 
In addition, we have 
\begin{equation*}
    \mathcal{L}_{\texttt{\textup{CADD}}}(T_{\texttt{\textup{RSCUSUM}}}) = \mathbb{E}_{1}[T_{\texttt{\textup{RSCUSUM}}}]-1.
\end{equation*}
Thus, we conclude that 
\begin{equation*}
    \mathcal{L}_{\texttt{\textup{CADD}}}(T_{\texttt{\textup{RSCUSUM}}})\sim \frac{\omega}{ \mathbb{D}_{\texttt{\textup{F}}}(P_1\|Q_{\infty})-\mathbb{D}_{\texttt{\textup{F}}}(P_1\|Q_{1})}.
\end{equation*}
Similar arguments applies for $\mathcal{L}_{\texttt{\textup{WADD}}}(T_{\texttt{\textup{RSCUSUM}}})$.
\end{proof}

\nocite{csiszar}

\section*{Acknowledgment}
Sean Moushegian, Suya Wu, and Vahid Tarokh were supported in part by Air
Force Research Lab Award under grant number FA-8750-
20-2-0504. Jie Ding was supported in part by the Office
of Naval Research under grant number N00014-21-1-2590.
Taposh Banerjee was supported in part by the U.S. Army
Research Lab under grant W911NF2120295.

\ifCLASSOPTIONcaptionsoff
  \newpage
\fi

\bibliographystyle{IEEEtran}
\bibliography{wu_571}

\balance

\begin{IEEEbiographynophoto}{Sean Moushegian}
received the B.S. degree in electrical engineering from Tufts University in 2021 and is currently pursuing the Ph.D. degree in electrical and computer engineering at Duke University.   His research interests include machine learning, sequential analysis, and information theory.
\end{IEEEbiographynophoto}

\begin{IEEEbiographynophoto}{Suya Wu} received the B.S. degree in mathematics and statistics from Shandong University, Jinan, China, in 2017, the M.S. degree in statistics from the University of Minnesota, Twin Cities, in 2019, and the Ph.D. degree in electrical and computer engineering from Duke University in 2023. Her research interests include statistical methods and machine learning. 
\end{IEEEbiographynophoto}

\begin{IEEEbiographynophoto}{Enmao Diao} was born in Chengdu, Sichuan, China, in 1994. He received the B.S. degree in computer science and electrical engineering from the Georgia Institute of Technology in 2016, the M.S. degree in electrical engineering from Harvard University in 2018, and the Ph.D. degree in electrical engineering from Duke University in 2023. His research interests include distributed machine learning, efficient machine learning, and signal processing.
\end{IEEEbiographynophoto}

\begin{IEEEbiographynophoto}{Jie Ding}
(Senior Member, IEEE) received the B.S. degree from Tsinghua University, Beijing, China, in 2012, and the Ph.D. degree in engineering sciences from Harvard University in 2017. In 2018, he joined as a Faculty Member with the University of Minnesota, Twin Cities, where he is currently an Associate Professor with the School of Statistics, with a graduate faculty appointment in Electrical Engineering and Computer Science. His research interests include data science foundations and various AI applications.
\end{IEEEbiographynophoto}

\begin{IEEEbiographynophoto}{Taposh Banerjee}
received the Ph.D. degree in electrical and computer engineering from the University of Illinois at Urbana–Champaign. He is currently an Assistant Professor of industrial engineering with the University of Pittsburgh. His research interests include sequential analysis, stochastic optimal control, and high-dimensional statistics and probability. 
\end{IEEEbiographynophoto}

\begin{IEEEbiographynophoto}{Vahid Tarokh}
(Fellow, IEEE) received the Ph.D. degree in electrical and computer engineering from the University of Waterloo, Waterloo, ON, Canada, in 1995. He worked at AT\&T Laboratories-Research, until 2000. From 2000 to 2002, he was an Associate Professor with the Massachusetts Institute of Technology (MIT). In 2002, he joined Harvard University as a Hammond Vinton Hayes Senior Fellow of Electrical Engineering and a Perkins Professor in applied mathematics. In January 2018, he joined Duke University as the Rhodes Family Professor of electrical and computer engineering. In 2018, he was also a Gordon Moore Distinguished Research Fellow with the California Institute of Technology. His current research interests include foundations and applications of machine learning, statistics, and dynamical systems.
\end{IEEEbiographynophoto}

\end{document}